\documentclass[journal]{IEEEtran} 
\IEEEoverridecommandlockouts

\usepackage[utf8]{inputenc}
\usepackage[T1]{fontenc}
\usepackage{Quant}
\usepackage{bm}

\usepackage[all=normal,paragraphs=tight,floats=normal,mathspacing=normal,wordspacing=tight,charwidths=tight,mathdisplays=normal,leading=normal]{savetrees}

\usepackage{graphicx}
\usepackage{xcolor}
\usepackage[bookmarks,colorlinks]{hyperref}
\usepackage{acronym, cite}
\usepackage{graphicx} 
\usepackage{paralist} 
\usepackage{float} 
\usepackage{amsmath}
\usepackage{amsthm}
\usepackage{tabularx,bbm}
\usepackage[ruled,linesnumbered,vlined]{algorithm2e}
\usepackage{enumitem}
\usepackage{algorithmic} 
\usepackage{amssymb}
\usepackage{amsmath}
\usepackage{adjustbox}
\usepackage{multirow}
\usepackage{pdfpages}
\usepackage[inkscapelatex=false]{svg}
\usepackage{ctable} 
\usepackage{colortbl} 
\usepackage{pifont}

\definecolor{NewColor}{rgb}{0,0,0}
\newtheorem{proposition}{Proposition}

\newcommand{\myVec}[1]{{\boldsymbol{#1}}}
\newcommand{\myMat}[1]{{\boldsymbol{#1}}}
\newcommand{\mySet}[1]{\mathcal{#1}}
 
\newcommand{\argmax}{\mathop{\mathrm{arg\,max}}\limits}
\newcommand{\Opt}{^{\rm opt}}
\newcommand{\Input}{\myVec{x}}
\newcommand{\InputSpace}{\mySet{X}} 
\newcommand{\Label}{\myVec{s}}
\newcommand{\LabelSpace}{\mySet{S}} 

\newcommand{\HypParam}{\myVec{\theta}^{\rm h}}

\newcommand{\PrecodedSymbol}{\myVec{p}}

\newcommand{\RPCAMat}{\myMat{X}}
\newcommand{\RPCARank}{\myMat{V}}
\newcommand{\RPCASparse}{\myMat{Y}}

\acrodef{dl}[DL]{deep learning}
\acrodef{ao}[AO]{alternating optimization}
\acrodef{cnn}[CNN]{Convolutional Neural Network}
\acrodef{dnn}[DNN]{deep neural network}
\acrodef{map}[MAP]{maximum a posteriori}
\acrodef{mlse}[MLSE]{maximum likelihood sequence estimate}
\acrodef{ml}[ML]{machine learning}
\acrodef{ss}[SS]{state space}
\acrodef{awgn}[AWGN]{additive white gaussian noise}
\acrodef{dnn}[DNN]{deep neural network} 
\acrodef{snr}[SNR]{signal-to-noise ratio}
\acrodef{cv}[CV]{constant velocity}
\acrodef{ca}[CA]{constant acceleration}
\acrodef{gnss}[GNSS]{global navigation satellite system}
\acrodef{pr}[PR]{pseudo range}
\acrodef{prr}[PRR]{pseudo range rate}
\acrodef{kg}[KG]{Kalman gain}
\acrodef{mb}[MB]{model based}
\acrodef{rnn}[RNN]{recurrent neural network}
\acrodef{mse}[MSE]{mean squared error}
\acrodef{ae}[AE]{Autoencoder}
\acrodef{mc}[MC]{Monte Carlo}
\acrodef{svd}[SVD]{singular value decomposition}
\acrodef{mimo}[MIMO]{multiple-input multiple-output}
\acrodef{pgd}[PGD]{Projected Gradient Descent}
\acrodef{pga}[PGA]{Projected Gradient Ascent}
\acrodef{bs}[BS]{base station}
\acrodef{rf}[RF]{Radio Frequency}
\acrodef{csi}[CSI]{channel state information}
\acrodef{flops}[FLOPS]{Floating Point Operations}
\acrodef{apga}[LAPGA]{Learned Approximated \ac{pga}}
\acrodef{dma}[DMA]{dynamic metasurface antenna}
\acrodef{dnns}[DNNs]{Deep Neural Networks}
\acrodef{rpca}[RPCA]{robust principal component analysis}
\acrodef{larpca}[LARPCA]{Learned Approximated \ac{rpca}}

\title{Deep Unfolding with Approximated\\ Computations for Rapid Optimization
}
\author{Dvir Avrahami$^*$, Amit Milstein$^*$, Caroline Chaux,~\IEEEmembership{Senior Member, IEEE},\\ Tirza Routtenberg,  \IEEEmembership{Senior Member, IEEE}, and Nir Shlezinger,  \IEEEmembership{Senior Member, IEEE}
\thanks{$^*$Equal contribution. 
Parts of this work were presented at the IEEE International Conference on Acoustics, Speech, and Signal Processing (ICASSP) 2025 as the paper \cite{milstein2025learned}. 
D. Avrahami, A. Milstein, T. Routtenberg, and N. Shlezinger are with the School of ECE, Ben-Gurion University of the Negev,  Israel (e-mail: \{dviravra; amitmils\}@post.bgu.ac.il, \{tirzar; nirshl\}@bgu.ac.il). 
C. Chaux is with CNRS, IPAL, Singapore (e-mail: caroline.chaux@cnrs.fr).  
This work was supported by the European Research Council (ERC) under the ERC Starting Grant No. 101163973 (FLAIR), by the ISRAEL SCIENCE FOUNDATION (Grant No. 1148/22), and by the Israeli Innovation Authority.
}
\vspace{-0.5cm}}
\begin{document}
%
\maketitle

\begin{abstract}
Optimization-based solvers play a central role in a wide range of signal processing and communication tasks. However, their applicability in latency-sensitive systems is limited by the sequential nature of iterative methods and the high computational cost per iteration. While deep unfolding has emerged as a powerful paradigm for converting iterative algorithms into learned models that operate with a fixed number of iterations, it does not inherently address the cost of each iteration. In this paper, we introduce a learned optimization framework that jointly tackles iteration count and per-iteration complexity. Our approach is based on unfolding a fixed number of optimization steps, replacing selected iterations with low-complexity approximated computations, and learning extended hyperparameters from data to compensate for the introduced approximations. We demonstrate the effectiveness of our method on two representative problems: $(i)$ hybrid beamforming;  and $(ii)$ robust principal component analysis. These fundamental case studies show that our learned approximated optimizers can achieve state-of-the-art performance while reducing computational complexity by over three orders of magnitude. Our results highlight the potential of our approach to enable rapid, interpretable, and efficient decision-making in real-time systems.
\end{abstract}

\acresetall

 \vspace{-0.2cm}
\section{Introduction} 

Iterative optimizers lie at the heart of numerous signal processing and communication tasks~\cite{luo2006introduction}. These methods provide principled solutions to tasks that can be mathematically formulated as optimization problems. 
However, various applications, ranging from wireless communications~\cite{shlezinger2024artificial} to real-time image processing~\cite{ahmadi2023fast} and edge computing~\cite{zhang2025optimization}, impose tight latency constraints that can render conventional iterative solvers impractical. In such settings, the time available to compute a solution is severely limited, and the use of classical optimization algorithms, which typically require tens to hundreds of iterations to converge, becomes a bottleneck. This pressing need for rapid and efficient solvers motivates the development of alternative approaches that can deliver high-quality solutions within tight computational budgets.

Traditionally, the performance of iterative optimization algorithms has been studied from a convergence perspective, focusing on their asymptotic behavior as the number of iterations grows~\cite{boyd2004convex}. In this classical view, the choice of algorithmic hyperparameters, such as step-sizes, penalty weights, or regularization strengths, is typically analyzed under assumptions that ensure theoretical convergence guarantees, often irrespective of their exact numerical values. However, in practice, these hyperparameters critically affect not only the convergence rate but also the computational efficiency of the solver. Techniques such as adaptive step-size selection or backtracking line search are commonly employed to accelerate convergence in terms of iterations~\cite[Ch. 3]{nocedal1999numerical}, \cite{cavalcanti2025adaptive}. Yet, these methods often increase per-iteration complexity, as they require evaluating multiple candidate updates, ultimately leading to increased overall latency, which is the very constraint that practical systems aim to minimize. Other methods include derivative-free optimization, where surrogate models and adaptive restart mechanisms are used to enhance robustness and efficiency under expensive or noisy function evaluations \cite{Coralia_2019_Lindon}.  However, these solvers typically require multiple sequential objective evaluations and do not address the per-iteration computational cost, as they rely on full-precision model evaluations. 

Recent years have witnessed a growing interest in leveraging deep learning tools to enhance optimization algorithms in various applications~\cite{dahrouj2021overview,Meiyi_Javad_2023,Xia_Beamforming_2020}, giving rise to the paradigm of \emph{learning to optimize}~\cite{chen2021learning}. Among the most structured and interpretable approaches within this paradigm is \emph{deep unfolding}~\cite{monga2019algorithm,shlezinger2025deep}, which systematically transforms iterative optimization algorithms into trainable neural network architectures. In deep unfolding, each iteration of the original solver is mapped to a layer in a feedforward network, with a common approach treating its algorithmic hyperparameters as learnable parameters~\cite{shlezinger2023model}. A key advantage of this approach lies in its ability to optimize the optimizer itself under a fixed computational budget, allowing one to preserve the interpretability and domain-specific structure of the original method while tuning its behavior to match the statistics of the data~\cite{shlezinger2022model}. This has been shown to enable reliable operation with a limited number of iterations of unfolded optimizers adopted in various domains, including wireless communications~\cite{balatsoukas2019deep, khani2020adaptive}, distributed optimization~\cite{noah2024distributed,saravanos2024deep}, beamforming~\cite{lavi2023learn, nguyen2023deep,balevi2021unfolded}, integrated sensing and communications~\cite{zhang2025deep, nguyen2024joint}, and outlier detection~\cite{cai2021learned,tan2023deep}, making the method highly suitable for latency-sensitive applications. Moreover, since the learning process is performed offline, no additional computational overhead is introduced at inference time, unlike adaptive schemes such as backtracking.

Nonetheless, in many applications, each iteration may still involve computationally expensive steps, such as matrix inversions, projections, or decompositions, that limit the achievable latency even when the number of iterations is small. This raises the question: can the ability to recast iterative solvers as neural networks be exploited not only to limit the number of iterations, but also to replace their most computationally intensive operations with lightweight, learnable approximations?

In this paper, we propose a new class of learned optimization solvers based on deep unfolding, which compensates both for a limited number of iterations and for approximated computations within each iteration. 
\textcolor{NewColor}{Beyond conventional deep unfolding, which primarily parametrizes iterative solvers~\cite{monga2019algorithm}, this work establishes deep unfolding as a principled framework for complexity reduction. This approach enables the deliberate incorporation of low-complexity approximations within iterative updates and compensating for their induced distortions through learnable parameterization.}
Our key insight is that leveraging data allows designing unfolded architectures that effectively absorb  the performance loss introduced by deliberately simplified operations. This enables designing solvers that are both fast, due to a fixed and small number of unfolded iterations, and lightweight, due to replacing computationally expensive steps with efficient surrogates. 
Unlike standard deep unfolding methods that focus on reducing iteration count via, e.g., learned hyperparameters, our framework introduces a second axis of approximation: replacing costly operations within each iteration with low-complexity surrogates. These approximations are compensated for by expanding the learnable parameter space, allowing the optimizer to adaptively absorb their effects. This dual approximation (in both iteration depth and per-iteration complexity) enhances our approach's applicability in latency- and resource-constrained environments.
In doing so, we uncover an additional core benefit of deep unfolding: its capacity to tolerate and learn around structural approximations, thereby reducing both the number and cost of iterations without sacrificing solution quality.

Our main contributions are summarized as follows:
\begin{itemize}
    \item \textbf{A new paradigm for learned approximated optimization:} We introduce a principled framework for constructing \emph{unfolded approximated iterative optimizers}, wherein selected steps of a classical iterative solver are replaced with lower-complexity operations. These approximate steps are embedded in a fixed-depth unfolded network whose parameters are optimized using training data, allowing the learned model to best match the data distribution. 
    
    \item \textbf{Case study I — Hybrid beamforming:} We apply our framework to the design of hybrid beamformers in multi-antenna wireless systems, where rapid optimization is critical~\cite{ahmed2018survey}. Starting from a classical alternating minimization algorithm, we construct an unfolded model with approximated update rules that avoid explicit matrix inversions. Our results demonstrate that the learned model achieves state-of-the-art beamforming performance while reducing computational complexity by over three orders of magnitude compared to the original solver.

    \item \textbf{Case study II — Robust principal component analysis:} We further demonstrate the versatility of our approach by applying it to \ac{rpca}~\cite{candes2011robust}, a fundamental problem in  signal processing. Using an unfolded version of the iterative thresholding algorithm with approximated gradient computations, we design a learned \ac{rpca} solver that maintains competitive recovery accuracy while dramatically reducing runtime.

    \item \textbf{Empirical validation across domains:} 
    We design and conduct comprehensive experiments that instantiate our framework to these two representative problems of hybrid beamforming and \ac{rpca}, each posing distinct structural and computational challenges. Through these case studies, we demonstrate that our methodology adapts to different optimization landscapes, quantifies the resulting trade-offs, and highlights that the learned approximations enable consistent,
     over three orders of magnitude reduction in runtime in both settings compared to conventional optimization while maintaining competitive performance.
\end{itemize}

The rest of this paper is organized as follows: Section~\ref{sec:system_model} details the considered optimization model, while Section~\ref{sec:Deep} formulates the generic methodology of learned approximated optimization solvers; Section~\ref{sec:HBF} specializes this methodology for hybrid beamforming, while Section~\ref{sec:PCA} details its case study for \ac{rpca}. Finally, Section~\ref{sec:conclusion} provides the concluding remarks.

Throughout this paper, we use boldface lower-case and upper-case letters for vectors (e.g., $\myVec{x}$) and matrices (e.g., $\myMat{X}$), respectively. The $(i,j)$th entry of  $\myMat{X}$ is denoted by $[\myMat{X}]_{i,j}$. Calligraphic letters denote sets, e.g., $\mySet{X}$, with $\mathbb{R}$, and $\mathbb{C}$ being the sets of real and complex numbers,  respectively, while $\|\cdot\|_F$, $\|\cdot\|_1$, $\|\cdot\|_2$, $\|\cdot\|_\ast$, $\odot$, and  $(\cdot)^T$  are the Frobenius/$\ell_1$/$\ell_2$/nuclear norm, Hadamard product, and transpose,  respectively.

 \vspace{-0.2cm}
\section{System Model and Preliminaries}\label{sec:system_model}
In this section, we formalize the optimization model behind our framework and outline the computational challenges it addresses. 
\vspace{-0.1cm}
\subsection{Optimization Model}
\label{sec:opt_model}
\vspace{-0.1cm}
We consider a generic decision-making setup that can be formulated as an optimization problem. While the following describes a general-purpose setup, we later specialize it to two representative case studies. 
Broadly speaking, the goal is to design a decision rule $f: \InputSpace \mapsto \LabelSpace$ that maps the context $\Input \in \InputSpace$, representing the available observations, into a decision $\Label \in \LabelSpace$. In optimization-based frameworks, the decision is obtained by solving a context-dependent minimization problem:
\begin{equation}
    \Label\Opt = \arg\min_{\Label \in \LabelSpace} \mathcal{L}_{\rm o}(\Label; \Input),
    \label{eq:opt_problem}
\end{equation}
where $\mathcal{L}_{\rm o}(\cdot\,;\cdot)$ is a task-dependent objective function encoding the loss associated with a given decision for a specific context.

In many practical scenarios, the optimization problem in \eqref{eq:opt_problem} is tackled using iterative algorithms. Such solvers proceed by repeatedly applying a context-aware update rule that gradually refines an initial estimate $\Label^{(0)}$. Iterative optimization methods often follow a per-iteration search principle, such as backtracking~\cite[Ch. 7]{boyd2004convex}, ensuring that the objective does not increase across iterations, i.e., that the estimate of each iteration $k$ holds
\begin{equation}
\label{eqn:descent}
    \mathcal{L}_{\rm o}(\Label^{(k+1)}; \Input) \leq \mathcal{L}_{\rm o}(\Label^{(k)}; \Input).
\end{equation}
Such descent methods are known to lead to convergence to a (local) minimum under appropriate assumptions, e.g., convexity of $ \mathcal{L}_{\rm o}$ in $\Label$ and convexity of the set $\LabelSpace$~\cite[Ch. 9]{boyd2004convex}.

Formally, the $k$th iterate is obtained via
\begin{equation}
    \Label^{(k+1)} = h_{\HypParam_k}(\Label^{(k)}; \Input),
    \label{eqn:IterativeMapping}
\end{equation}
where $h_{\HypParam_k}: \LabelSpace \times \InputSpace \mapsto \LabelSpace$ denotes the update function applied at iteration $k$. This update is governed by solver-specific parameters $\HypParam_k$ (e.g.,  step-sizes or momentum factors), referred to as  \emph{hyperparameters}, which guide the update steps but do not alter the optimization objective itself.   While often fixed to constant values, which are either hand-tuned or derived via theoretical convergence criteria, hyperparameters play a critical role in determining the convergence speed, runtime and {overall solution quality} of the optimizer.

\vspace{-0.1cm}
\subsection{Challenges of Rapid Optimization}
\label{sec:rapid_opt_challenges}
\vspace{-0.1cm}
In many applications, optimization-based decision-making must be performed under tight latency constraints. These settings impose a fundamental challenge: while optimization solvers provide principled and accurate solutions, their iterative nature often renders them impractical for rapid deployment.

Classical iterative optimization methods face three main limitations when applied to time-sensitive decision problems:
\begin{enumerate}[label={\em C\arabic*},series=challenges]
    \item \label{itm:ManyIter} \textbf{Large number of iterations:} Reaching a solution that is sufficiently close to the optimizer of \eqref{eq:opt_problem} often requires a considerable number of iterations. The precise number of steps needed to achieve a target level of accuracy can vary significantly across contexts, and is typically not known in advance. Consequently, worst-case iteration budgets should be accounted for, potentially exceeding the permissible time budget in real-time systems.

    \item \label{itm:SeqComp} \textbf{Sequential computation:} Iterative solvers operate sequentially as in \eqref{eqn:IterativeMapping}, with each update depending on the result of the previous one. This inherent structure precludes parallelization across iterations, limiting the ability to exploit modern computational resources for acceleration and making latency tightly coupled to the number of steps executed.

    \item \label{itm:ComplexIter} \textbf{High per-iteration complexity:} Most descent methods rely on differentiating the objective function, typically through gradient computations. In many applications, such as large-scale matrix factorization or beamforming, the gradient expressions involve complex operations (e.g., matrix multiplications, inversions, or projections) that are computationally intensive. As a result, each iteration, i.e., the computation of $ h_{\HypParam_k}$ in  \eqref{eqn:IterativeMapping}, may be slow to compute, even if the total number of iterations is modest.
\end{enumerate}

Taken together, \ref{itm:ManyIter}-\ref{itm:ComplexIter} highlight the key bottlenecks of classical optimization solvers in latency-sensitive environments. They motivate the need for novel methodologies that can retain the principled nature of optimization-based decision-making while achieving fast inference through architectural and computational simplifications. In the following, we propose a learned approximated optimization framework that addresses these challenges by leveraging end-to-end deep learning tools while preserving the interpretability of iterative solvers.

 \vspace{-0.2cm}
\section{Learned Approximated Optimization}\label{sec:Deep} 

This section introduces our proposed methodology for \emph{learned approximated optimization}, which aims to overcome the limitations highlighted in~\ref{itm:ManyIter}--\ref{itm:ComplexIter} while preserving the interpretability and principled operation of classical iterative solvers. The approach is based on designing a trainable architecture that operates with a fixed number of iterations, where each iteration may involve low-complexity approximated computations. In the following, we first present the rationale underlying our approach (Subsection \ref{subsec:rationale}), then detail the proposed methodology (Subsection \ref{subsec:methodology}), and finally provide a discussion highlighting its scope, flexibility, and implementation aspects (Subsection \ref{ssec:discussion}).

\vspace{-0.1cm}
\subsection{Rationale}
\label{subsec:rationale}
\vspace{-0.1cm}
As noted in~\ref{itm:ManyIter} and~\ref{itm:SeqComp}, a major bottleneck in applying iterative optimization solvers to latency-sensitive applications stems from the variable and potentially large number of iterations required for convergence. To support rapid decision-making, it is thus essential to design optimization algorithms that operate with a fixed and relatively small number of iterations~$K$, chosen in advance. A key enabler for such fixed-depth operation is the framework of deep unfolding, which has been shown to preserve the interpretable structure of classical solvers while learning iteration-specific hyperparameters from data~\cite{shlezinger2022model}. In this paradigm, the iterative mapping~\eqref{eqn:IterativeMapping} is viewed as a $K$-layer neural network, with each layer corresponding to a single iteration whose parameters are tuned offline for the target domain. 
\textcolor{NewColor}{Interpretability here follows from the optimization-based structure, as each layer corresponds to an optimization step and produces an intermediate estimate of the decision variables. Consequently, the internal representations exchanged between layers correspond to progressively refined estimates of $\Label$.}
By casting the hyperparameters $\{\HypParam_k\}_{k=1}^K$ as the trainable parameters of a discriminative machine learning model~\cite{shlezinger2022discriminative}, one fully retains the interpretable operation of the original solver while achieving reliable performance within a small number of steps.

However, limiting the number of iterations alone does not fully resolve the latency bottleneck. As discussed in~\ref{itm:ComplexIter}, the complexity of each iteration may itself be prohibitive, particularly when involving costly gradient evaluations or matrix operations. To address this, we take a further step by replacing selected iterations with low-complexity approximated computations, designed to significantly reduce the per-iteration runtime. Of course, such approximations may degrade performance if left uncompensated. To maintain the reliability of the resulting optimization process, we extend the hyperparameter space of the unfolded solver in a way that does not increase computational burden (e.g., by using element-wise or structure-aware parameterizations), and we train the full approximated optimization procedure end-to-end using data. In this manner, we obtain an optimization model that is both interpretable and efficient, and can achieve strong performance with low latency and low computational overhead.

\vspace{-0.1cm}
\subsection{Methodology}
\label{subsec:methodology}
\vspace{-0.1cm}
\subsubsection{Increased Parameterization}
Our proposed inference rule is based on an unfolded and approximated iterative optimizer. We begin by fixing the number of iterations to a small, predetermined value $K$, enabling bounded and predictable runtime suitable for real-time applications. In contrast to classical solvers that often utilize shared or scalar hyperparameters across iterations, we increase the expressiveness of the optimizer by introducing \emph{iteration-specific}, \emph{extended hyperparameters}. We denote the full set of hyperparameters for iteration $k$ by $\Theta_k$. 

An important design choice in our methodology lies in how we parameterize the learnable hyperparameters $\Theta_k$ at each iteration. A key insight is that increasing the expressiveness of these hyperparameters need not come at the cost of increased computational complexity.  For instance, a parameterization employed throughout our case studies replaces scalar step-sizes with element-wise vector-valued parameters that enable finer-grained control in gradient-based optimization. As we show in the following running example (as well as in the case studies in Sections~\ref{sec:HBF}-\ref{sec:PCA}), such abstraction can be designed to avoid increasing complexity, and preserve the interpretable operation of such optimizers as descent methods.

{\bf Example (Gradient Descent).} 
As a running example for our methodology, we consider a basic first-order optimizer based on a standard gradient-based update rule:
\begin{equation}
\Label^{(k+1)} =  h_{\HypParam_k}(\Label^{(k)}; \Input) = \Label^{(k)} - \eta_k \cdot \nabla_{\Label} \mathcal{L}_{\rm o}(\Label^{(k)}; \Input),
\label{eqn:ElementWiseGD_1}
\end{equation}
where $\eta_k$ is a scalar step-size, i.e., $\HypParam_k = \eta_k$. Replacing this scalar with a multivariate step-size $\boldsymbol{\eta}_k$ of the same shape as $\Label$, namely, $\Theta_k = \boldsymbol{\eta}_k$, yields:
\begin{equation}
\Label^{(k+1)} =  h_{\Theta_k}(\Label^{(k)}; \Input) =  \Label^{(k)} - \bm{\eta}_k \odot \nabla_{\Label} \mathcal{L}_{\rm o}(\Label^{(k)}; \Input).
\label{eqn:ElementWiseGD}
\end{equation}
 This formulation increases the flexibility in learning iteration-specific dynamics, without increasing the number of floating-point operations, as the element-wise product is of the same computational order as scalar multiplication. This approach is reminiscent of variable metric algorithms with diagonal preconditioners. 
 In addition, it also preserves the semantics of descent-based optimization under mild regularity assumptions, as stated in the following proposition.
\begin{proposition}
\label{prop:elementwise_descent}
Let $\mathcal{L}_{\rm o}(\Label; \Input)$ be a differentiable objective function, and let  $\bm{\eta}_k$ be comprised of positive step-sizes.
Then the update \eqref{eqn:ElementWiseGD} yields a decrease in the objective value as in \eqref{eqn:descent} to first order, i.e., the method retains the descent property under sufficiently small element-wise step-sizes.
\end{proposition}
\begin{proof}
By the first-order Taylor approximation:
\begin{align*} 
&\mathcal{L}_{\rm o}(\Label^{(k+1)}; \Input)
\stackrel{(a)}{\approx} \mathcal{L}_{\rm o}(\Label^{(k)};\Input) \!+\! \nabla_{\Label} \mathcal{L}_{\rm o}(\Label^{(k)};\Input)^T (\Label^{(k+1)} \!- \!\Label^{(k)}) \nonumber \\
&\stackrel{(b)}{=}  \mathcal{L}_{\rm o}(\Label^{(k)};\Input) - \nabla_{\Label} \mathcal{L}_{\rm o}(\Label^{(k)};\Input)^T {\rm diag} (\bm{\eta}_k)   \nabla_{\Label} \mathcal{L}_{\rm o}(\Label^{(k)};\Input),
\end{align*}
where the approximation in $(a)$ holds for sufficiently small step-sizes, and $(b)$ is obtained by substituting the update step from \eqref{eqn:ElementWiseGD} and rearranging, with ${\rm diag}(\bm{\eta}_k)$ being a diagonal matrix whose diagonal entries are $\bm{\eta}_k$. 
Since these entries are positive, then this diagonal matrix is positive-definite, and thus, $\mathcal{L}_{\rm o}(\Label^{(k+1)};\Input) \leq \mathcal{L}_{\rm o}(\Label^{(k)};\Input)$, i.e., \eqref{eqn:ElementWiseGD} is a descent method.
\end{proof}
\textcolor{NewColor}{
Proposition~\ref{prop:elementwise_descent} demonstrates that, using our running example of gradient descent, increased abstractness through element-wise parameterization can be made to $(i)$ preserve computational complexity; $(ii)$ not compromise theoretical guarantees such as descent behavior; and $(iii)$ provide richer representational power to absorb approximations. It does not provide finite-iteration performance guarantees; rather, performance within a fixed number of iterations $K$ is achieved via the data-driven training procedure detailed in the sequel.
}

\subsubsection{Constructing the Approximated Iterative Optimizer}
\label{sssec:Theo2}
To reduce the computational burden of the resulting optimizer, we select a subset of the iteration indices, denoted by $\mathcal{K}^{\rm approx} \subseteq \{1, 2, \ldots, K\}\triangleq \mySet{K}$. In these iterations, we replace the standard update rule $h_{\Theta_k}$ from \eqref{eqn:IterativeMapping} with a low-complexity approximated computation, denoted $\hat{h}_{\Theta_k}$. The resulting decision rule is recursively given by the output of the $K$th iteration, with 
\begin{equation}
    \label{eqn:approxOpt}
     \Label^{(k+1)} = 
     \begin{cases}
         \hat{h}_{\Theta_k}(\Label^{(k)}; \Input) & k \in \mathcal{K}^{\rm approx}, \\
         {h}_{\Theta_k}(\Label^{(k)}; \Input) & k \notin \mathcal{K}^{\rm approx}.
     \end{cases}
\end{equation}

The approximated mappings are designed to mimic the behavior of their full-precision counterparts while avoiding the need for expensive operations such as matrix inversions, gradient evaluations, or projections. In the subsequent case studies, we provide concrete examples of such approximations, including skip updates, momentum-based surrogates, and fixed-pattern substitutions. Specifically, for the running example of gradient descent, we next prove that such approximations can be made to have only a mild effect on the error bounds of the original optimizers.

{\bf Example (Gradient Descent).} 
To illustrate, consider again the gradient-based update rule of \eqref{eqn:ElementWiseGD_1}, where for a set of iterations the objective gradients $\nabla \mathcal{L}_{\rm o}(\Label^{(k)}; \Input)$ are replaced with an approximated computation, denoted $\tilde{\myVec{g}}^{(k)}(\Label^{(k)}; \Input)$, namely:
\begin{equation}
\label{eqn:selective_update}
    \Label^{(k+1)} = \Label^{(k)} - 
    \begin{cases}
       \bm{\eta}_k \odot \tilde{\myVec{g}}^{(k)}(\Label^{(k)}; \Input) &  k \in \mathcal{K}^{\rm approx}, \\
        \bm{\eta}_k \odot  \nabla \mathcal{L}_{\rm o}(\Label^{(k)}; \Input) & \text{otherwise}.
    \end{cases}
\end{equation} 
Accordingly, $\hat{h}_{\Theta_k}(\Label^{(k)}; \Input) =  \Label^{(k)} - \bm{\eta}_k \odot \tilde{\myVec{g}}^{(k)}(\Label^{(k)}; \Input)$. 
Under conventional assumptions used in analysis of gradient-based optimization~\cite{garrigos2023handbook}, we can rigorously characterize the impact of such selective approximations on the error of the optimizer, as stated in the following proposition.


\begin{proposition}
\label{prop:approx_gd_partial}
Assume that $\mathcal{L}_{\rm o}(\Label;\Input)$ is $\mu$-strongly convex and $L$-smooth with $L \geq \mu > 0$.
Consider the recursion in \eqref{eqn:selective_update}, 
and let the step-size vector $\bm{\eta}_k$ satisfy $0< \underline{\eta}_k\le [\bm{\eta}_k]_i \le \bar{\eta}_k < \frac{2}{L}$ for all $i$, 
and that for every $k \in \mathcal{K}^{\rm approx}$,
\begin{equation}
\left\| {\bm{\eta}}_k \odot \left(\tilde{\myVec{g}}^{(k)}(\Label^{(k)};\Input) - \nabla \mathcal{L}_{\rm o}(\Label^{(k)};\Input)\right) \right\|_2 \leq \delta_k ,
\label{eq:error_delta}
\end{equation}
while for $k \notin \mathcal{K}^{\rm approx}$ we set $\delta_k=0$.
Then, for any minimizer $\Label^\star$ of $\mathcal{L}_{\rm o}(\cdot;\Input)$, the objective after $K$ steps satisfies:
\begin{align}
\mathcal{L}_{\rm o}(\Label^{(K)};\Input)\!-\!\mathcal{L}_{\rm o}(\Label^\star;\Input)
\leq
&(1\!-\! \mu c)^K\!\left(\mathcal{L}_{\rm o}(\Label^{(0)};\Input)\!-\!\mathcal{L}_{\rm o}(\Label^\star;\Input)\right) \notag \\
&+\sum_{k=0}^{K-1}(1-\mu c)^{K-k-1} C\,\delta_k^2 ,
\label{eqn:bound_partial}
\end{align}
where $c \triangleq \inf_k \underline{\eta}_k\!\left(1-\frac{L}{2}\sup_k \bar{\eta}_k\right)$, $C \triangleq \frac{L}{2} + \frac{(1+L\sup_k \bar{\eta}_k)^2}{2c}$.
\end{proposition}

\begin{proof}
   The proof is detailed in Appendix~\ref{app:proof1}. 
\end{proof}

\textcolor{NewColor}{Proposition~\ref{prop:approx_gd_partial} characterizes how approximation errors accumulate in the gradient-descent example and how their impact is modulated by the step-size parameters. Specifically,}
 the first term in \eqref{eqn:bound_partial} corresponds to the standard convergence rate of gradient descent~\cite[Thm. 3.4]{garrigos2023handbook}.
The second term in \eqref{eqn:bound_partial} accumulates errors that are based on the deviations from the true gradient only over the selected iterations $\mathcal{K}^{\rm approx}$.
This result shows that the final error scales with the magnitude of the skipped gradients and the approximation error $\delta_k$. The proposition indicates the potential usefulness of deliberately inducing approximated computations (for complexity reduction) with increased parameterization, since as long as the skipped gradients (balanced by the increased hyperparameters $\bm \eta _k$) are small, the error remains bounded and the method remains effective.  \textcolor{NewColor}{Proposition~\ref{prop:approx_gd_partial}  provides theoretical insight into the interplay between approximation and parameterization, rather than guaranteeing bounded error for arbitrary approximation levels; using deep unfolding techniques, the hyperparameters are learned from data to compensate for approximation-induced distortions, as detailed in the following subsection.}

\subsubsection{Training via Empirical Risk Minimization}\label{ssec:training_setups}
To learn the parameters $\Theta \triangleq \{\Theta_k\}_{k=1}^K$ of the approximated  optimizer, we assume access to a representative dataset $\mathcal{D}$. The  training procedure depends on the nature of the available supervision.

\smallskip
    \textbf{Unsupervised setting:} When the dataset is given by $\mathcal{D} = \{\Input_t\}$, and no ground-truth decisions are available, we guide the training process using the same objective function $\mathcal{L}_{\rm o}$ used in the original optimization formulation. In this case, the empirical risk is defined as
    \begin{equation}
        \mathcal{L}_{\mathcal{D}}(\Theta) = \frac{1}{|\mathcal{D}|} \sum_{\Input_t \in \mathcal{D}} \mathcal{L}_{\rm o} \left( \hat{\Label}(\Input_t | \Theta); \Input_t \right),
        \label{eqn:lossUnsup}
    \end{equation}
    where $\hat{\Label}(\Input_t | \Theta)$ denotes the output of the unfolded approximated optimizer with parameters $\Theta$ applied to input $\Input_t$.

\smallskip
    \textbf{Supervised setting:} When the dataset contains both contexts and target labels, i.e., $\mathcal{D} = \{(\Input_t, \Label_t)\}$, we are not restricted to the proxy optimization objective $\mathcal{L}_{\rm o}$, and can directly train the optimizer to minimize a task-oriented loss that evaluates the quality of the predicted decision relative to the target label. In this case, the empirical risk is defined as
    \begin{equation}
        \mathcal{L}_{\mathcal{D}}(\Theta) = \frac{1}{|\mathcal{D}|} \sum_{(\Input_t, \Label_t) \in \mathcal{D}} \mathcal{L}_{\rm task} \left( \hat{\Label}(\Input_t | \Theta), \Label_t \right),
        \label{eqn:losssup}
    \end{equation}
    where $\mathcal{L}_{\rm task}(\cdot, \cdot)$ is a supervised loss function  chosen according to the target application. Supervised learning is most useful in settings where  the optimization objective $\mathcal{L}_{\rm o}$ is a mathematically formulated proxy of the actual system task.

\smallskip
In both cases, the unfolded approximated optimizer is trained offline using standard gradient-based methods to minimize the empirical risk over $\Theta$. This allows the model to compensate for the approximations in both iteration count and complexity, yielding an efficient and reliable learned decision rule. 
\textcolor{NewColor}{The ability of empirical risk minimization to compensate for the performance loss introduced by computational approximations stems from two main factors: $(i)$ the increased parameterization compared to the original iterative solver, which provides additional degrees of freedom that enable the learned optimizer to minimize the influence of approximation error terms in the error-bound expansion (as shown in Proposition~\ref{prop:approx_gd_partial} for gradient descent); and $(ii)$ the parameters are learned end-to-end via empirical risk minimization, where the loss evaluates the quality of the final output of the $K$-layer architecture, This allows the training process to identify parameter configurations that learn the implicit bias of the data distribution, effectively 'weighting' the update steps to correct distortions in the intermediate update trajectory caused by the approximated computations.}
The formulation of training an approximated unfolded optimizer via conventional deep learning methods, i.e., mini-batch \ac{sgd}, is stated in Algorithm~\ref{alg:APGA_train}.

\begin{algorithm}
    \caption{Training approximated unfolded optimizer via mini-batch \ac{sgd}}
    \label{alg:APGA_train} 
    \SetKwInOut{Initialization}{Init}
    \Initialization{Select iterative optimizer \eqref{eqn:IterativeMapping}\;
                   $\quad$ Fix iterations $K$ and approximated iterations $\mySet{K}^{\rm approx}$\;
                   $\quad$ Extend hyperparameters $\Theta$ and init as fixed\; 
                   $\quad$ Set learning rate $\rho$}
    \SetKwInOut{Input}{Input} 
    \Input{Training set  $\mathcal{D}$ (labeled or unlabeled)}   
    {
        \For{${\rm epoch} = 0, 1, \ldots, {\rm epoch}_{\max}-1$}{%
                    Randomly divide  $\mathcal{D}$ into $Q$ batches $\{\mathcal{D}_q\}_{q=1}^Q$;
                    
                    \For{$q = 1, \ldots, Q$}{
  
                    Compute  empirical risk $\mySet{L}_{\mySet{D}_q}$ via \eqref{eqn:lossUnsup} or \eqref{eqn:losssup};
                    
                    Update  $\Theta\leftarrow \Theta - \rho\nabla_{\Theta}\mathcal{L}_{\mySet{D}_q}(\Theta)$; \label{stp:update1}
                    }
                    
                    }
        \KwRet{$\Theta$}
  }
\end{algorithm}

\vspace{-0.6cm}
\subsection{Discussion}\label{ssec:discussion}
\vspace{-0.1cm}
The proposed methodology introduces a novel and flexible approach to accelerate iterative optimization by combining deep unfolding with approximated computations. While conventional uses of deep unfolding aim to reduce the number of iterations by learning solver hyperparameters from data, our framework extends this idea by deliberately introducing low-complexity approximations within the iterations themselves. These approximations are typically regarded as limitations that degrade performance; however, we show that when integrated into an unfolded architecture and paired with appropriate parameterization, the resulting model can be trained to absorb and compensate for the induced mismatches. In doing so, our methodology directly addresses all three challenges \ref{itm:ManyIter}-\ref{itm:ComplexIter}. This dual approximation on both iteration depth and per-iteration complexity is made tractable and effective through the trainability offered by deep unfolding. The increased abstractness in adding hyperparameters is designed not to come at the cost of increased complexity. Moreover, different abstractions, e.g., the usage of a per-parameter step-size rather than scalar ones as we do in our case studies, in fact preserve the property of descent methods as in \eqref{eqn:descent}~\cite{khobahi2021lord}.

\textcolor{NewColor}{Like most learning-based methods, unfolded optimizers may require retraining when the system configuration changes, particularly when their parameters depend on the data size (as is the case with extended hyperparameters). However, due to their structured origin in iterative optimization algorithms, unfolded architectures often exhibit stronger scalability properties than generic \acp{dnn}. In particular, as the trainable parameters correspond to solver hyperparameters, these parameters can often be extended across problem dimensions through simple scaling or replication strategies~\cite{noah2024distributed} to enable zero-shot generalization to larger problem instances. Alternatively, the modular structure of unfolded optimizers can be combined with hypernetworks that generate hyperparameters conditioned on the system configuration as proposed in \cite{raviv2025modular}, enabling elastic architectures capable of adapting to varying problem dimensions without retraining. These extensions are left for future study.}

\textcolor{NewColor}{The proposed framework is intended as a design methodology rather than an off-the-shelf `plug-and-play' algorithm.  Because the sources of computational complexity vary significantly across different iterative solvers  applications, the choice of which operations to approximate and how to approximate them must  be tailored on a case-by-case basis. Accordingly, many of the practical considerations and benefits of our framework are revealed only when applied to specific use cases.}  In the following sections, we illustrate how our general methodology can be instantiated in two representative signal processing tasks. These case studies demonstrate how approximations can be selected, quantify their computational impact, and empirically validate the ability of the learned optimizer to preserve or even improve performance.



 \vspace{-0.2cm}
\section{Case Study 1: Hybrid Beamforming}
\label{sec:HBF}
This section presents an application of our methodology within an unsupervised learning setup, as introduced in Subsection~\ref{ssec:training_setups}, focusing on the task of {\em hybrid beamforming}~\cite{shlezinger2024artificial}, which is considered key in enabling high-frequency large-scale \ac{mimo} systems~\cite{elbir2022twenty}. Our motivation for using hybrid beamforming as a case study stems not only from its technological importance, but also from its need for rapid optimization solvers, since: 
$(i)$ The beamforming task is typically represented as an optimization problem that is tackled with iterative solvers~\cite{yu2016alternating,sohrabi2016hybrid,gong2019rf,qiao2020alternating}; 
$(ii)$ Hybrid beamformers are tuned for a given \ac{csi},  which can change rapidly (over $10^3$ times a second~\cite{shlezinger2024artificial});
$(iii)$ While unfolded optimizers have been shown to yield suitable hybrid beamformers within a few iterations~\cite{lavi2023learn,nguyen2023deep,balevi2021unfolded,levy2025rapid,shi2022deep},  each iteration is often computationally intense, especially in wideband regimes, such that the optimizer cannot be applied within a coherence duration. 
Accordingly, this case study represents an optimization setup where there is a concrete need to simultaneously reduce the number of iterations and alleviate the burden within the iterations. \textcolor{NewColor}{While this case study focuses on alternating optimization based on \ac{pga}, selected as a representative baseline due to its widespread use, its established empirical performance, and its suitability for deep unfolding~\cite{lavi2023learn}, the proposed methodology is generic and can be applied to alternative hybrid beamforming optimizers as well. Our emphasis is not on the specific choice of solver, but on demonstrating how iterative algorithms can be transformed into reduced-complexity, trainable models via approximated deep unfolding.}

To present the case study, we first formulate the task in Subsection~\ref{ssec:HBF_model}, followed by the description of the \ac{pga} optimizer for this task in Subsection~\ref{ssec:HBF_opt}. We explain how the proposed methodology is incorporated into the resulting algorithm and discuss its performance in Subsections~\ref{ssec:HBF_Lopt} and~\ref{ssec:HBF_Exp}, respectively.

\vspace{-0.1cm}
\subsection{Hybrid Beamforming 
Formulation}\label{ssec:HBF_model}
\vspace{-0.1cm}
\subsubsection{Channel Model}
Consider a single-cell downlink \ac{mimo} system featuring an $M$-antenna \ac{bs} and $N$ users. 
The \ac{bs} utilizes $B$ frequency bands, shared among all users in a non-orthogonal fashion. 
Let $\mathbf{s}_b \in \mathbb{C}^N$ denote the symbol vector transmitted at the $b$-th frequency bin, where $b \in \mathcal{B} \triangleq \{1, 2, \dots, B\}$. 
The symbols are assumed to be i.i.d. with equal power such that $\mathbb{E}[\mathbf{s}_b \mathbf{s}_b^H] = \frac{1}{N} \mathbf{I}_N$ for each $b \in \mathcal{B}$.

We precode the symbols $\mathbf{s}_b$ into the channel input $\PrecodedSymbol_b \in \mathbb{C}^M$, yielding the channel output:
	\begin{equation}
	\label{eqn:io relation}
	\myVec{y}_b = \myMat{H}_b \cdot \PrecodedSymbol_b + \myVec{n}_b \in\mathbb{C}^N.
	\end{equation}
Where $\myMat{H}_b\in\mathbb{C}^{N\times M}$ is the channel at the $b$th band, and $\myVec{n}_b\in\mathbb{C}^N$ is Gaussian noise with i.i.d. entries of variance $\sigma^2$. 

\subsubsection{Hybrid Beamforming}
With $L < M$ \ac{rf} chains at the \ac{bs}, a two-stage hybrid precoding scheme is employed for the symbols $\{\mathbf{s}_b\}_{b \in \mathcal{B}}$. 
In the first stage, a frequency-dependent digital precoder $\mathbf{W}_{d,b} \in \mathbb{C}^{L \times N}$ is applied. 
In the second stage, an analog precoder $\mathbf{W}_a \in \mathcal{A}$ combines these signals into the channel input $\PrecodedSymbol_b$. 
Following~\cite{shlezinger2024artificial}, $\mathbf{W}_a$ is frequency-invariant and consists of phase shifters restricted to the unit-magnitude set $\mathcal{A}$, namely:
\begin{equation}
	\label{eqn:ph_shift}
   \myMat{W}_a\in \mySet{A} = \big\{\myMat{A}\in \mathbb{C}^{M\times{L}}\big| |[\myMat{A}]_{m,l}| = 1, \quad \forall (m,l) \big\}.
\end{equation}

The precoded channel input for the $b$-th frequency bin is thus
\begin{equation}
    \PrecodedSymbol_b = \mathbf{W}_a \mathbf{W}_{d,b} \mathbf{s}_b.
\end{equation}
Based on \eqref{eqn:io relation}, the  channel input-output relationship is 
\begin{equation}
\label{eqn:b received signal}
\mathbf{y}_b = \mathbf{H}_b \mathbf{W}_a \mathbf{W}_{d,b} \mathbf{s}_b + \mathbf{n}_b, \quad \forall b \in \mathcal{B}.
\end{equation}
Considering the normalized symbol power, the hybrid precoders are subject to an average power constraint across all frequency bins, which is formulated as
\begin{equation}
\label{eqn:Power}
 \frac{1}{B} \sum_{j=1}^{B} \|\mathbf{W}_{a}\mathbf{W}_{d,j} \|^2_F \leq N.
\end{equation}

\subsubsection{Problem Formulation}\label{ssec:problem_formulation}

The hybrid beamforming design objective is to determine the precoders $\mathbf{W}_a$ and $\{\mathbf{W}_{d,b}\}$ for a given set of \ac{csi} $\{\mathbf{H}_b\}$. Namely, the {\em decision} here is $\Label = [\myMat{W}_a,\{\myMat{W}_{d,b}\}]$, and the {\em context} is $\Input = \{\myMat{H}_b\}$, as illustrated in Fig.~\ref{fig:HBF}.
Due to its unsupervised setting, we utilize the objective function of the original optimization problem as our empirical risk as discussed in Section~\ref{ssec:training_setups}. Specifically, we focus on maximizing the sum-rate, a standard performance metric that is invariant of receiver-side processing. 
Assuming the \ac{bs} possesses perfect \ac{csi}, the achievable sum-rate is given by~\cite{shen2007sum}:
	\begin{align}
	&{R}\left(\myMat{W}_a , \{\myMat{W}_{d,b}\}_{b\in \mySet{B}}; \{\myMat{H}_b\}_{b\in \mySet{B}}\right) 
	 \notag  \\ & = \frac{1}{B} \sum_{b=1}^{B} \log \left|\textbf{I}_N \!+\! \frac{1}{N \sigma^2} \myMat{H}_b\myMat{W}_a\myMat{W}_{d,b} \myMat{W}_{d,b}^H\myMat{W}_a^H\myMat{H}_b^H\right|. \label{eqn:achievable rate}
	\end{align}
Consequently, the task of determining the hybrid precoders can be cast as the following optimization problem:
\begin{align}
\label{eqn:main problem}
&\argmax_{\myMat{W}_a\in \mySet{A} , \{\myMat{W}_{d,b}\}_{b\in \mySet{B}}:~~ {\rm s.t.}~~\eqref{eqn:Power}} R\left(\myMat{W}_a , \{\myMat{W}_{d,b}\}, \{\myMat{H}_b\}\right). 
\end{align}
For $\mySet{S} = \{\myMat{W}_a\in \mySet{A} , \{\myMat{W}_{d,b}\}_{b\in \mySet{B}}: {\rm s.t.}~\eqref{eqn:Power}\}$ and $\mathcal{L}_{\rm o}(\Label; \Input) = - R(\Label; \Input)$, Eq.~\eqref{eqn:main problem} specializes the generic formulation in \eqref{eq:opt_problem}.

\vspace{-0.1cm}
\subsection{Iterative Optimizer}\label{ssec:HBF_opt}
\vspace{-0.1cm}
The constrained maximization problem \eqref{eqn:main problem} can be addressed, among other techniques, using \ac{pga} combined with alternating optimization~\cite{lavi2023learn}. 
In this framework, each iteration alternates between optimizing the analog precoder $\mathbf{W}_a$ with fixed $\{\mathbf{W}_{d,b}\}$, and subsequently updating the digital precoders while $\mathbf{W}_a$ remains constant. 
Each update is followed by a projection step to ensure that $\mathbf{W}_a \in \mathcal{A}$ and the power constraint in \eqref{eqn:Power} are satisfied. 

Specifically, the $k$-th iteration updates $\mathbf{W}_a$ via
\begin{subequations}
    \label{eqn:PGAsteps}
    \begin{align}
	{\myMat{W}_a^{(k+1)}}
	= &\Pi_{\mySet{A}}\Big\{\myMat{W}_a^{(k)}
	  \notag \\&+\mu_a^{(k)}\frac{\partial}{\partial\myMat{W}_a}
	{R}(\myMat{W}_a^{(k)}, \{\myMat{W}_{d,b}^{(k)}\}, \{\myMat{H}_b\})\Big\},\label{eqn:Wa_step_proj}
    \end{align}
while each digital precoder  $\myMat{W}_{d,b}$ is updated as
 	\begin{align}
	\hspace{-0.1cm} {\myMat{W}}_{d,b}^{(k+1)} 
	= &\Pi_{\mySet{P}}\Big\{\myMat{W}_{d,b}^{(k)}
  \notag \\&+ \mu_{d,b}^{(k)}\frac{\partial}{\partial\myMat{W}_{d,b}}
	{R}\Big(\myMat{W}_a^{(k\!+\!1)} , \{\myMat{W}_{d,b}^{(k)}\}, \{\myMat{H}_b\} \Big)\Big\}.  \hspace{-0.25cm} 	\label{eqn:gr_Wdb_k_step}  
	\end{align}
\end{subequations} 
The hyperparameters  $\HypParam_k = [\mu_a^{(k)},\{\mu_{d,b}^{(k)}\}]$ are scalar positive step-sizes, while $\Pi_{\mySet{A}}$ and $\Pi_{\mySet{P}}$ are projection operators asserting that  $\myMat{W}_a \in \mySet{A}$ and \eqref{eqn:Power} respectively hold. 
As shown in \cite{lavi2023learn}, by defining $\tilde{\myMat{H}}_b \triangleq \sqrt{\frac{1}{N \sigma^2}}\myMat{H}_b$ and $\myMat{G}_b(\myMat{W}_a,\myMat{W}_{d,b}, \myMat{H}_b) \triangleq  (\myMat{I}_N +  \tilde{\myMat{H}}_b\myMat{W}_a\myMat{W}_{d,b} \myMat{W}_{d,b}^{H}\myMat{W}_a^H \tilde{\myMat{H}}_b^H)$, the rate gradients 
are given by
\begin{subequations}
 \label{eqn:Gradients}
\begin{align}
	&\hspace{-0.3cm}\frac{\partial}{\partial\myMat{W}_a}
	{R}(\myMat{W}_a  , \{\myMat{W}_{d,b}\}, \{\myMat{H}_b\} ) 
	 \nonumber \\ \hspace{-0.8cm}=&\frac{1}{B} \sum_{b=1}^{B}\tilde{\myMat{H}}_b^T\textbf{G}_b(\myMat{W}_a ,\myMat{W}_{d,b} , \myMat{H}_b)^{-T} \tilde{\myMat{H}}_b^*
	\myMat{W}_a^{*} \myMat{W}_{d,b}^{*} \myMat{W}_{d,b}^T,\label{eqn:gr_Wa_k_calc}  
 \end{align}
 and
  	\begin{align}
	&\frac{\partial  }{\partial\myMat{W}_{d,b}}	{R}(\myMat{W}_a, \{\myMat{W}_{d,b}\}, \{\myMat{H}_b\})  \nonumber \\ &= \frac{1}{B}\myMat{W}_a^T\tilde{\myMat{H}}_b^T\textbf{G}_b(\myMat{W}_a,\myMat{W}_{d,b}, \myMat{H}_b)^{-T}\tilde{\myMat{H}}_b^* \myMat{W}_a^{*} \myMat{W}_{d,b}^{*}. \label{eqn:gr_Wd_k_calc}
	\end{align}
 \end{subequations}
 \ac{pga} sets hybrid precoders by repeating \eqref{eqn:PGAsteps} until convergence.

\begin{figure}
    \centering
    \includegraphics[width=\columnwidth]{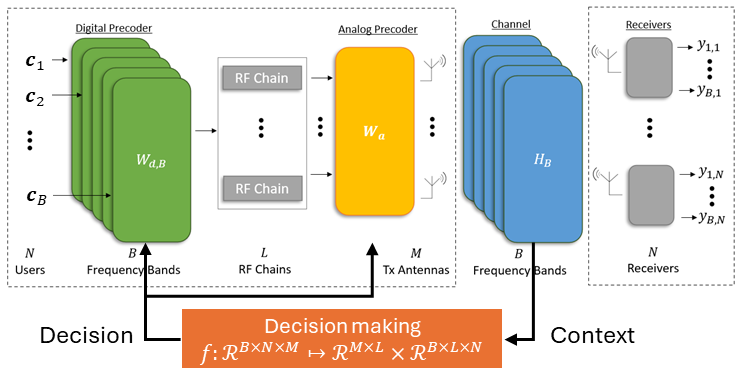}
    \vspace{-0.4cm}
    \caption{Hybrid beamforming case study illustration}
    \label{fig:HBF}
     \vspace{-0.4cm}
\end{figure}

\vspace{-0.25cm}
\subsection{Learned Approximated Optimizer}\label{ssec:HBF_Lopt}
\vspace{-0.1cm}
While \ac{pga} can yield suitable hybrid precoders, it is often impractical for real-time operation within a channel coherence duration. 
Specifically, the Hybrid Beamforming task is subject to the core limitations of classical iterative methods discussed in Section~\ref{ssec:problem_formulation};
it requires a large number of iterations to converge (\ref{itm:ManyIter}), and suffers from high per-iteration complexity (\ref{itm:ComplexIter}) due to gradient calculations. 
Following~\cite{lavi2023learn}, the computational cost per iteration is given by
\begin{equation}
    \mathcal{C}_{\rm PGA}^{\rm iter} \propto 2B \times (NML + N^3 + M^2L + L^2N).
    \label{eqn:Complexity}
\end{equation}

Following the methodology outlined in Subsection~\ref{subsec:methodology}, we propose \ac{apga} to solve \eqref{eqn:main problem}. 
\ac{apga} leverages data-driven parameterization to replace standard gradient computations with low-complexity approximations, thereby maintaining precoder performance while preserving the iterative structure of \ac{pga}. 
Specifically, we address \ref{itm:ManyIter} by fixing the number of iterations to a small integer $K$. 
Furthermore, to mitigate \ref{itm:ComplexIter}, we reduce the per-iteration complexity by employing the following approximations:
\begin{enumerate}[label={\em I\arabic*},series=innovations] \label{innovations}
  \item \label{innov_wa_approx} \textbf{$\frac{\partial{R}}{\partial\myMat{W}_a}$ Approximation} - 
  \
  For a subset of iterations, $\mySet{K}_{a}^{\rm approx}$, we replace   the gradient in  \eqref{eqn:gr_Wa_k_calc}, with a fixed matrix, using $\mathbf{1}_{M\times L}$, i.e., an all-ones   $M \times L$ matrix. 
  Accordingly, in the $k$th iteration we approximate $\frac{\partial{R}}{\partial\myMat{W}_a}$ as 
\begin{equation}
\label{eqn:WaApprox}
    \frac{\partial{\tilde{R}^{(k)}}}{\partial\myMat{W}_a} = \begin{cases}	    	
\mathbf{1}_{M\times L}, & k \in \mySet{K}_a^{\rm approx} \\
 \eqref{eqn:gr_Wa_k_calc} & k \notin \mySet{K}_a^{\rm approx}.
 \end{cases}
\end{equation} 
{{The all-ones matrix $\mathbf{1}_{M \times L}$ is used as a surrogate gradient direction because it provides a uniform positive update across all entries, and allows expressive learning when combined with element-wise step-sizes (see \ref{innov_stepping_size_matrix} below). Thus, it introduces a low-complexity yet trainable update mechanism that maintains flexibility across all dimensions.}}
  \item \label{innov_wd_approx} \textbf{$\frac{\partial{R}}{\partial\myMat{W}_{d,b}}$ Approximation} - 
  Here, we fix a set of per-band iterations $\mySet{K}_{d,b}^{\rm approx}$ for each $b \in \mySet{B}$. For these iterations, \ac{apga} uses the gradient of the previous iteration as a form of momentum.   Accordingly, in the $k$th iteration we approximate $\frac{\partial{R}}{\partial\myMat{W}_{d,b}}$ as $\frac{\partial \tilde{{R}}_b^{(k)}}{\partial\myMat{W}_{d,b}}$, taken  via \eqref{eqn:gr_Wd_k_calc} with 

  \begin{equation*}
      \tilde{{R}}_b^{(k)} =
    \begin{cases}
	{R}(\myMat{W}_a^{(k+1)} , \{\myMat{W}_{d,\tilde{b}}^{(k)}\}, \{\myMat{H}_{\tilde{b}}\}) & k \notin \mySet{K}_{d,b}^{\rm approx},\\
	{R}(\myMat{W}_a^{(k)} , \{\myMat{W}_{d,\tilde{b}}^{(k-1)}\}, \{\myMat{H}_{\tilde{b}}\})& k \in \mySet{K}_{d,b}^{\rm approx}.
    \end{cases} 
  \end{equation*}
\end{enumerate}
Approximated computations introduced in \ref{innov_wa_approx}-\ref{innov_wd_approx} implement \eqref{eqn:approxOpt}, with $\mySet{K}^{\rm approx} =\big(\cup_b \mySet{K}_{d.b}^{\rm Approx} \big) \cup \mySet{K}_a^{\rm Approx}$. 

As detailed in Subsection~\ref{subsec:methodology}, we also increase the parameterization. In \ac{apga}, this is reflected in the following aspect:
 \begin{enumerate}[label={\em I\arabic*},resume=innovations] 
   \item \label{innov_stepping_size_matrix}\textbf{Element-Wise Step-Sizes} - To provide additional degrees of freedom for coping with errors induced by approximations, we replace the scalar step-sizes $\HypParam_k$ with matrix ones, such that $\Theta_k = [\bm{\mu}_a^{(k)} , \{\bm{\mu}_{d,b}^{(k)}\}]$, with $\bm{\mu}_a^{(k)} \in \mathbb{R}^{M \times L}$  and $\bm{\mu}_{d,b}^{(k)} \in \mathbb{R}^{L \times N}$.  This abstraction does not increase complexity as it utilizes per-entry step-sizes, while reducing to scalar step-sizes by setting $\bm{\mu}_a^{(k)} =  {\mu}_a^{(k)} \myMat{1}_{M \times L}$ and $\bm{\mu}_{d,b}^{(k)}  = {\mu}_{d,b}^{(k)} \myMat{1}_{L \times N}$. 
\end{enumerate}

\textcolor{NewColor}{
We note that numerous combinations of approximations, as in \ref{innov_wa_approx}-\ref{innov_wd_approx}, are possible. The specific configuration in \ac{apga} was selected as it represents a robust operating point for the complexity-performance balance. This choice leverages the unit-magnitude constraints of the analog array and the temporal stability of digital precoding iterations to maintain sum-rate while minimizing computational overhead. However, the approximation setup can be tailored to meet the specific requirements of different deployment scenarios.}
Combining \ref{innov_wa_approx}-\ref{innov_stepping_size_matrix}, the resulting \ac{apga} (using the initiation proposed in \cite{lavi2023learn}) is summarized as Algorithm~\ref{alg:APGA_inference}. \ac{apga} is trained using  data  $\mySet{D}$ comprised of past channel realizations, written as 
$\mySet{D}=\{\{\tilde{\myMat{H}}_b^r\}_{b\in\mySet{B}}\}_{r=1}^{|\mathcal{D}|}$, i.e., there is no "ground-truth" precoders. Accordingly, training is carried out based on the unsupervised empirical risk~\eqref{eqn:lossUnsup}, using, e.g.,  Algorithm~\ref{alg:APGA_train}.   


\begin{algorithm}
    \caption{\ac{apga} for Hybrid Precoding}
    \label{alg:APGA_inference}
    \SetKwInOut{Initialization}{Init}
    \Initialization{Step-sizes $\{\bm{\mu}_{d,b}^{(k)}\}_{b\in \mySet{B}},\bm{\mu}_a^{(k)}$;  
     $\{\myMat{W}_{d,b}^{(0)}\}_{b\in \mySet{B}}$; \newline
    Approximation indices $\mySet{K}_a^{\rm approx}$, $\{\mySet{K}_{{d,}b}^{\rm approx}\}_{b\in\mySet{B}}$;   }
    \SetKwInOut{Input}{Input} 
    \Input{Channel matrices $\{\tilde{\myMat{H}}_b\}_{b\in \mySet{B}}$}  
    {
        $\myMat{W}_a^{(0)} \leftarrow$   $L$ right-singular vectors of $\frac{1}{B} \sum_b\tilde{\myMat{H}}_b$\;
        \For{$k = 0, 1, \ldots K-1$}{ 
                    
                     $\myMat{W}_a^{(k+1)} \leftarrow\Pi_{\mySet{A}}\left\{\myMat{W}_a^{(k)}
 	\!+ \!\bm{\mu}_a^{(k)}\odot     \frac{\partial{\tilde{R}^{(k)}}}{\partial\myMat{W}_a} \right\}$ ; \label{line:P_Wa}
                     
                    \For{$b = 1, \ldots, B$}{%
                     $\myMat{W}_{d,b}^{(k\!+\!1)} \!\leftarrow\Pi_{\mySet{P}}\left\{\myMat{W}_{{a}}^{(k)}
 	\!+ \!\bm{\mu}_{d,b}^{(k)}\odot     \frac{\partial{\tilde{R}_b^{(k)}}}{\partial\myMat{W}_{d,b}} \right\}$;  \label{line:P_Wd}
                    }
                
                }
        \KwRet{$\myMat{W}_{d,b}^{(K)}$, $\{\myMat{W}_{d,b}^{(K)}\}_{b\in \mySet{B}}$}
  }
\end{algorithm}

\ac{apga} implements an iterative optimizer characterized by a small number of iterations and low-complexity, approximated operations. By leveraging data-driven insights, it mitigates potential performance loss while significantly reducing computational overhead. Specifically, the average complexity per iteration is:
\begin{equation}
    \mathcal{C}_{\rm LAPGA}^{\rm iter} \approx \mathcal{C}_{\rm PGA}^{\rm iter} \cdot \left( 1 - \frac{ |\mySet{K}_a^{\rm approx}|}{2K} - \frac{\sum_{b\in \mySet{B}}|\mySet{K}_{d,b}^{\rm approx}|}{2B\cdot K} \right),
    \label{Eqn:CompRed}
\end{equation} 
which follows since the complexity in \eqref{eqn:Complexity} is obtained by equal contributions of the digital and analog  settings~\cite{lavi2023learn}.
The reduction in \eqref{Eqn:CompRed} is shown in Section~\ref{ssec:HBF_Exp} to have a significant impact on latency with minimal impact on performance. 

\vspace{-0.1cm}
\subsection{Numerical Evaluation}\label{ssec:HBF_Exp}
\vspace{-0.1cm}
In this section, we evaluate the performance of \ac{apga} through two numerical simulations. Our primary objective is to demonstrate that the proposed framework achieves performance comparable to non-approximated solvers while significantly reducing the online computational burden. To this end, we compare \ac{apga} against the following benchmarks, selected for their direct relationship to the \ac{pga} framework and their established efficacy in hybrid beamforming~\cite{shlezinger2024artificial}:
\begin{enumerate}[label={\em B\arabic*}]  
    \item \label{bnchmrk:mu_scalar} \textbf{Unfolded Full PGA:} The learned architecture from \cite{lavi2023learn} using exact gradients and trainable scalar step-sizes.
    \item \label{bnchmrk:classic_pga} \textbf{Classic PGA:} The traditional iterative solver using manually optimized, fixed step-sizes. 
\end{enumerate}

We first consider a small-scale \ac{mimo} configuration with $B=8$ frequency bins, $N=6$ users, $L=10$ RF chains, and $M=12$ antennas. 
The channel realizations are generated using the QuadRiGa model~\cite{jaeckel2014quadriga}, utilizing $1,000$ instances for training and $100$ for testing. 

We implement \ac{apga} by employing the approximations in \ref{innov_wa_approx}--\ref{innov_stepping_size_matrix}. 
Specifically, we set $\mathcal{K}_a^{\rm approx} = \{1,\dots,K\}$ and define $\mathcal{K}_{d,b}^{\rm approx}$ as the set of even iteration indices. 
This setup is further utilized to assess the individual impact of the proposed approximations. 
To this end, we compare \ac{apga} against the benchmarks \ref{bnchmrk:mu_scalar}--\ref{bnchmrk:classic_pga} as well as the following architectural variations:
\begin{enumerate}[label={\em V\arabic*}] 
    \item \label{var:mu_matrix_approx_dwa} $\mySet{K}_{d,b}^{\rm approx} \equiv \emptyset$, i.e.,  only \ref{innov_wa_approx}   and \ref{innov_stepping_size_matrix}.
    \item \label{var:mu_matrix} $\mySet{K}_{d,b}^{\rm approx} \equiv \emptyset$ and $\mySet{K}_a^{\rm approx} = \emptyset$, i.e.,  only  \ref{innov_stepping_size_matrix}.
    \item \label{var:mu_scalar_approx_dwa_approx_dwd} Scalar learned step-sizes, i.e.,  only  \ref{innov_wa_approx} and \ref{innov_wd_approx}.
    \item \label{var:mu_scalar_approx_dwa} Scalar learned step-sizes and $\mySet{K}_{d,b}^{\rm approx} \equiv \emptyset$, i.e.,  only  \ref{innov_wa_approx}.
\end{enumerate}

\begin{figure}
  \centering
  \vspace{-0.8cm}
  \includegraphics[width=0.5\textwidth]{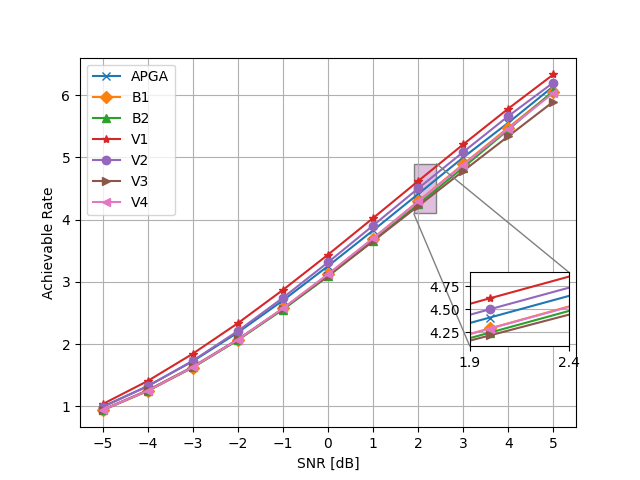}
  \vspace{-0.5cm}
  \caption{Rate vs \acs{snr}, small-scale \ac{mimo}.}
  \label{fig:small_scale}
\end{figure}

\begin{table}
\centering
 \fontsize{7.5pt}{10pt}\selectfont
\begin{tabular}{|c|c|c|c|c|c|c|}
\hline
 \cellcolor{lightgray}
{\textbf{Method}}&  \cellcolor{lightgray} \bm{$K$} &  \cellcolor{lightgray}{\ref{innov_wa_approx}}& \cellcolor{lightgray} \ref{innov_wd_approx}&  \cellcolor{lightgray}\ref{innov_stepping_size_matrix} & \cellcolor{lightgray} \textbf{Mean Sum-Rate} &  \cellcolor{lightgray}\textbf{\# Products}\\
\hline
\ac{apga}& 5 &\checkmark & \checkmark & \checkmark  &3.25 $\pm$ 0.167& 71,424  \\
\hline
\ref{bnchmrk:mu_scalar}& 5 &\ding{55} & \ding{55} & \ding{55} & 3.12 $\pm$ 0.177& 238,080   \\
\hline
\ref{bnchmrk:classic_pga}& 50 &\ding{55} & \ding{55} & \ding{55} & 3.08 $\pm$ 0.169 & 2,380,800 \\
\hline
\ref{var:mu_matrix_approx_dwa}& 5 &\checkmark & \ding{55} & \checkmark & 3.43 $\pm$ 0.157& 119,040\\
\hline
\ref{var:mu_matrix}& 5 &\ding{55} & \ding{55} & \checkmark & 3.31 $\pm$ 0.178& 238,080 \\
\hline
\ref{var:mu_scalar_approx_dwa_approx_dwd}& 5 &\checkmark & \checkmark & \ding{55} & 3.09 $\pm$ 0.165& 71,424 \\
\hline
\ref{var:mu_scalar_approx_dwa}& 5 &\checkmark &\ding{55}  & \ding{55} & 3.12 $\pm$ 0.179& 119,040 \\
\hline

\end{tabular}
\vspace{0.2cm}
\caption{
Comparison of achievable rates and computational complexity
 for $0$ dB \acs{snr} in small-scale \ac{mimo}.}
\label{table:methods_summary_small}
\vspace{-0.5cm}
\end{table}

Fig.~\ref{fig:small_scale} reports the average per-user sum-rate versus \ac{snr} (defined as $1/N\sigma^2$) for $K=5$ iterations, except for \ref{bnchmrk:classic_pga}, which uses $K=50$. 
We observe in Fig.~\ref{fig:small_scale} that deep unfolding leads to significant performance gains. 
Specifically, \ref{bnchmrk:mu_scalar} achieves better results than \ref{bnchmrk:classic_pga} while requiring $90\%$ fewer iterations, while \ac{apga} surpasses the rate obtained by \ref{bnchmrk:mu_scalar} despite requiring $70\%$ fewer calculations. 
Overall, our proposed algorithm outperforms \ref{bnchmrk:classic_pga} while reducing the computational burden by $89\%$--$97.3\%$, depending on the \ac{snr}.

To further quantify these gains, Table~\ref{table:methods_summary_small} summarizes the sum-rate (mean and standard deviation) at $0$ dB and the total number of real-valued products required for the precoding task. 
The table highlights the dramatic complexity reduction of \ac{apga}. 
It also demonstrates that the approximations \ref{innov_wa_approx}--\ref{innov_wd_approx} allow for various balances between performance and complexity, and that the expanded parameterization in \ref{innov_stepping_size_matrix} is crucial for learning to mitigate the errors induced by these approximations. 

Next, we evaluate the proposed framework in a larger \ac{mimo} setting, with $B=64$, $N=12$, $L=12$, and $M=32$. 
Building on the ablation study conducted in the previous section, we now consider \ac{apga} employing the full suite of approximations in \ref{innov_wa_approx}--\ref{innov_stepping_size_matrix}. 
We maintain $K=5$ iterations for the unfolded optimizers, whereas \ref{bnchmrk:classic_pga} is configured with $K = 100$ iterations to ensure convergence.
\begin{figure}
  \centering 
  \vspace{-0.7cm}
  \includegraphics[width=0.45\textwidth,trim={0.5cm 0 0.5cm 0.5cm},clip]{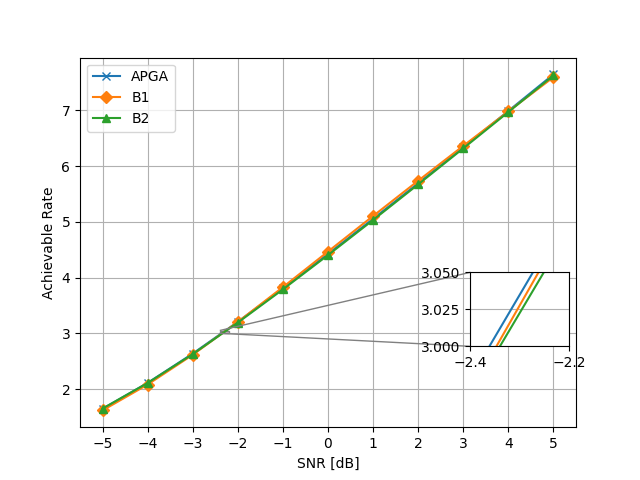}
  \vspace{-0.2cm}
  \caption{Rate vs. \ac{snr}, large-scale \ac{mimo}.}
  \label{fig:large_scale}
\end{figure}
The average per-user sum-rate versus \ac{snr} is illustrated in Fig.~\ref{fig:large_scale}, while Table~\ref{table:methods_summary_large} provides a detailed complexity analysis at $0$ dB. 
The results indicate that \ac{apga} achieves superior performance in both the low and high \ac{snr} regimes compared to the benchmarks. 
In the mid-range \ac{snr}, while \ref{bnchmrk:mu_scalar} exhibits slightly higher rates, it entails a $70\%$ increase in computational overhead. 
Furthermore, \ref{bnchmrk:classic_pga} is shown to be three orders of magnitude more complex than our proposed approach. 


\begin{table}
\centering
 \fontsize{7.5pt}{10pt}\selectfont
\begin{tabular}{|c|c|c|c|c|c|c|}
\hline
 \cellcolor{lightgray}
{\textbf{Method}} &  \cellcolor{lightgray} \bm{$K$} & \cellcolor{lightgray} {\ref{innov_wa_approx}}&  \cellcolor{lightgray} \ref{innov_wd_approx}&  \cellcolor{lightgray} \ref{innov_stepping_size_matrix} &  \cellcolor{lightgray} \textbf{Mean Sum-Rate} &   \cellcolor{lightgray} \textbf{\# Products}\\
\hline
\ac{apga}& 5 &\checkmark & \checkmark & \checkmark &  4.42 $\pm$ 0.092& 3,907,584 \\
\hline
\ref{bnchmrk:mu_scalar}& 5 &\ding{55} & \ding{55} & \ding{55} & 4.46 $\pm$ 0.110 & 13,025,280   \\
\hline
\ref{bnchmrk:classic_pga}& 100 &\ding{55} & \ding{55} & \ding{55} & 4.40 $\pm$ 0.092 &260,505,600  \\
\hline
\end{tabular}
\vspace{0.2cm}
\caption{Comparison of achievable rates and computational complexity
for $0$ dB \ac{snr} in large-scale \ac{mimo}.}
\label{table:methods_summary_large}
\vspace{-0.25cm}
\end{table} 

 \vspace{-0.2cm}
\section{Case Study 2: Robust PCA}
\label{sec:PCA}

{{Through the hybrid beamforming study}}, we highlighted how unfolding combined with approximated computations can enable efficient inference with drastically reduced latency. 
In this section, we explore the application of \ac{rpca}, which is a widely-used task in various fields ranging from background subtraction and anomaly detection in image and video processing~\cite{bouwmans2018applications}. \ac{rpca} serves as a representative setup for optimization problems that are 
$(i)$ computationally demanding and have motivated several unfolded, yet non-approximated, optimization strategies in prior work~\cite{cai2021learned,tan2023deep};
$(ii)$ formulated using surrogate mathematical objectives rather than the true task loss. In particular, conventional \ac{rpca} formulations rely on convex relaxations of the underlying objective~\cite{candes2011robust},  which may be misaligned with the ultimate performance metric of interest. This makes \ac{rpca} an especially suitable domain for examining the added value of our approach in scenarios where the learning objective deviates from the optimization one.

To describe this case study, we first formulate \ac{rpca} and the considered first-order optimizer in Subsections~\ref{ssec:PCA_model} and \ref{ssec:PCA_opt}, respectively. Then we detail our unfolded methodology and evaluate it in Subsections~\ref{ssec:PCA_Lopt} and \ref{ssec:PCA_Exp}, respectively.

\vspace{-0.1cm}
\subsection{\ac{rpca} Formulation}\label{ssec:PCA_model}
\vspace{-0.1cm}
\subsubsection{Task}
\ac{rpca} addresses the decomposition of a data matrix $\RPCAMat\in\mathbb{R}^{n_1\times n_2}$ into a low-rank matrix $\RPCARank$ and a sparse matrix $\RPCASparse$, capturing structured and sparse components, respectively. Formally, \ac{rpca} is based on the assumption that
\begin{equation}
    \RPCAMat = \RPCARank^\star + \RPCASparse^\star,
    \label{eqn:RPCAModel}
\end{equation}
where $\RPCARank^\star$ captures the low-rank structure, and $\RPCASparse^\star$ represents the sparse component. 
The \ac{rpca} problem refers to the recovery of the low-rank and sparse components from the data matrix, i.e., the {\em decision} is $\Label = \{\RPCARank, \RPCASparse\}$, which should approach the true $\RPCARank^\star$ and $\RPCASparse^\star$, from the {\em context} $\Input = \RPCAMat$.

\subsubsection{Optimization Objective}
There are different convex approaches for solving the \ac{rpca} task (see, e.g. \cite{candes2011robust}).

To reduce complexity and computational cost, several non-convex formulations were proposed. Here, we adopt the one proposed in \cite{yi2016fast}, which factors the low-rank matrix as $\RPCARank = \myMat{L}\myMat{R}^{\top}$, where $\myMat{L}\in\mathbb{R}^{n_1\times r}$ and $\myMat{R}\in\mathbb{R}^{n_2\times r}$, with $r\ll\min(n_1,n_2)$.
Then, it solves the non-convex  optimization problem:
\begin{align}
\label{eqn:RPCA_nonconvex}
&\min_{\myMat{L},\myMat{R},\RPCASparse}  \mathcal{L}_{\rm o}(\myMat{L},\myMat{R},\RPCASparse; \RPCAMat) \triangleq \frac{1}{2}\|\myMat{L}\myMat{R}^{\top}\!+\!\RPCASparse\!-\!\RPCAMat\|_F^2, \\
\notag 
&
\quad\text{s.t. }\text{supp}(\RPCASparse)\subseteq\text{supp}(\RPCASparse^\star),
\end{align}
where $\text{supp}(\RPCASparse)$ denotes the support (indices of non-zero entries) of the sparse matrix $\RPCASparse$. While the formulation in \eqref{eqn:RPCA_nonconvex} requires prior knowledge of the unknown support of $\RPCASparse^\star$, it gives rise to an efficient iterative solver detailed in the sequel, which operates without such knowledge. 



\vspace{-0.1cm}
\subsection{Iterative Optimizer}\label{ssec:PCA_opt} 
\vspace{-0.1cm}
Problem~\eqref{eqn:RPCA_nonconvex} can be tackled by using an iterative algorithm employing gradient steps and soft-thresholding projection. We base our iterative optimizer on the method proposed in~\cite{cai2021learned}, which alternates between updating the sparse component $\RPCASparse$ and the low-rank components $\myMat{L},\myMat{R}$.

Specifically, the $k$th iteration updates the estimate of $\RPCASparse$ as
\begin{equation}\label{eq:outlier_update}
  \RPCASparse^{(k+1)} = \mySet{T}_{\zeta^{(k+1)}}\!\bigl(
      \RPCAMat-\myMat{L}^{(k)}(\myMat{R}^{(k)})^T
  \bigr),
\end{equation}
where $\zeta^{(k+1)}$ is a positive threshold hyperparameter, and  $\mySet{T}_{\zeta}$ is the element-wise soft-thresholding operator, given by
\begin{equation}
[\mySet{T}_{\zeta}(\myMat{M})]_{i,j} = \text{sign}([\myMat{M}]_{i,j})\max(|[\myMat{M}]_{i,j}| - \zeta,0).
\label{eq:soft_thresh}
\end{equation}

The updated matrices $\myMat{L},\myMat{R}$ using scaled gradient steps are
\begin{align}
\myMat{L}^{(k+1)}  
  &=\myMat{L}^{(k)} - \myVec{\eta}_L^{(k+1)} \odot \bigl(\nabla_{\myMat{L}}  \mathcal{L}_{\rm o}(\myMat{L}^{(k)},\myMat{R}^{(k)},\RPCASparse^{(k)}; \RPCAMat) \notag \\
            &\qquad\qquad \times ((\myMat{R}^{(k)})^T\myMat{R}^{(k)})^{-1}\bigr), 
\label{eq:update_L}
\\
\myMat{R}^{(k+1)} 
&= \myMat{R}^{(k)} 
   - \myVec{\eta}_R^{(k+1)} \odot \bigl(\nabla_{\myMat{R}} \mathcal{L}_{\rm o}(\myMat{L}^{(k)},\myMat{R}^{(k)},\RPCASparse^{(k)}; \RPCAMat) \notag \\
            &\qquad\qquad \times ((\myMat{L}^{(k+1)})^T\myMat{L}^{(k+1)})^{-1}\bigr), 
\label{eq:update_R}
\end{align}
where $\myVec{\eta}_L^{(k+1)}={\eta}_L^{(k+1)}\myVec{1},  \myVec{\eta}_R^{(k+1)}={\eta}_R^{(k+1)}\myVec{1} $ are the {\em{scalar}} step-sizes hyperparameters. 
The  objective gradients are obtained as
\begin{equation}
\label{eq:grad_L}
\nabla_{\myMat{L}}  \mathcal{L}_{\rm o}(\myMat{L},\myMat{R},\RPCASparse; \RPCAMat)
=  
  \bigl(\myMat{L} \myMat{R}^T  + \RPCASparse - \RPCAMat\bigr)\myMat{R}, 
\end{equation}
\begin{equation}
\label{eq:grad_R}
\nabla_{\myMat{R}}  \mathcal{L}_{\rm o}(\myMat{L},\myMat{R},\RPCASparse; \RPCAMat)
= 
  \bigl(\myMat{L}\myMat{R}^T  + \RPCASparse - \RPCAMat\bigr)^T \myMat{L}. 
\end{equation}
The iterations commence by setting $\RPCASparse^{(0)} = \mySet{T}_{\zeta^{(0)}}(\RPCAMat)$ and initializing ${\eta}_L^{(0)}$ and ${\eta}_R^{(0)} $ from the truncated \ac{svd} of $\RPCAMat-\RPCASparse^{(0)}$.

\vspace{-0.1cm}
\subsection{Learned Approximated Optimizer}\label{ssec:PCA_Lopt}
\vspace{-0.1cm}
\subsubsection{Unfolding RPCA}
While the \ac{rpca} based on \eqref{eq:outlier_update}-\eqref{eq:update_R} can yield suitable offline processing, in many scenarios it is impractical for real-time applications. This stems from the fact that 
$(i)$ it requires multiple iterations to converge; and $(ii)$ each iteration involves heavy computations, particularly when the rank of $\RPCARank$ is not too small. Specifically, the complexity per iteration can be broken down into the following components\footnote{For brevity, we take $n \triangleq n_1 = n_2$ when discussing  complexities}: 
\begin{enumerate}
    \item \textbf{Sparse component update \eqref{eq:outlier_update}}:  This step involves computing the residual and applying element-wise soft-thresholding, which requires $n^2 r + n^2$ flops.
    \item \textbf{Low-rank component updates \eqref{eq:update_L}--\eqref{eq:update_R}}: Let $\mathcal{P}_{\mathrm{low}}$ denotes the computational cost of updating a single factor ($\myMat{L}$ or $\myMat{R}$). Each low-rank update involves computing the gradient and a normalization factor, which requires an inversion of an $r\times r$ matrix. The total cost per update therefore satisfies $\mathcal{P}_{\mathrm{low}} \propto n^2 r + 2nr^2 + r^3$ flops. 
\end{enumerate}
    Since we need to update both $\myMat{L}$ and $\myMat{R}$ in each iteration, the total cost of computing the low‑rank is  $2\mathcal{P}_{\mathrm{low}}$ flops.
In total, the complexity per iteration  is:
\begin{equation}
    \mathcal{C}_{\rm RPCA}^{\rm iter}  = n^{2}r + n^{2} +2\mathcal{P}_{\mathrm{low}}.
    \label{eqn:rpca_Complexity}
\end{equation}

The work \cite{cai2021learned} proposed to facilitate operation with a limited number of iterations by treating the thresholds $\{\zeta^{(k)}\}$ and the step-sizes $\{\eta^{(k)}\}$ (which were fixed to be identical for both $\myMat{L}$ and $\myMat{R}$) as learnable parameters of an unfolded network.  This approach enables efficient operation with a fixed number of iterations. However, the per-iteration complexity 
remains unchanged and is still given by
\eqref{eqn:rpca_Complexity}. 
Notably, the successful application of deep unfolding to these iterative updates motivates further reducing complexity by integrating deliberate approximation based on our methodology.

\subsubsection{Learned Approximated \ac{rpca}}
Our methodology leads to {\em \ac{larpca}}, which replaces some of the gradient computations with low-complexity approximations, combined with increased parameterization and learned hyperparameters to preserve the operation of \ac{rpca}. Accordingly, we first fix the number of iterations to a small integer $K$. To further reduce the per-iteration complexity,  we   introduce the following approximations:
\begin{enumerate}[label={\em J\arabic*},series=innovations2] \label{approxi}
  \item \label{L_approx} \textbf{$\nabla_{\myMat{L}}  \mathcal{L}_{\rm o}$ Approximation} - 
  \
  For a subset of iterations $\mySet{K}_{L}^{\rm approx}$, instead of computing the full gradient in  \eqref{eq:grad_L}, \ac{larpca} reuses the gradient of the previous iteration as a form of momentum, i.e., 
\begin{equation*}
      \myMat{L}^{(k+1)} =
    \begin{cases}
	\myMat{L}^{(k)} & k \in \mySet{K}_{L}^{\rm approx} ,\\
        \eqref{eq:update_L} & k \notin \mySet{K}_{L}^{\rm approx} .\\
    \end{cases} 
  \end{equation*}
  \item \label{R_approx} \textbf{$\nabla_{\myMat{R}}  \mathcal{L}_{\rm o}$ Approximation} - Similarly, for a subset of iterations $\mySet{K}_{R}^{\rm approx} $,  the update of $\myMat{R}$ is skipped and its value is retained from the previous iteration:    
\begin{equation*}
      \myMat{R}^{(k+1)} =
    \begin{cases}
	\myMat{R}^{(k)} & k \in \mySet{K}_{R}^{\rm approx} ,\\
        \eqref{eq:update_R} & k \notin \mySet{K}_{R}^{\rm approx} .\\
    \end{cases} 
  \end{equation*}
\end{enumerate}
The approximated computations in \ref{L_approx}-\ref{R_approx} are instances of \eqref{eqn:approxOpt}, with $\mySet{K}^{\rm approx} = \mySet{K}_{L}^{\rm approx}  \cup \mySet{K}_R^{\rm approx}$. 

As discussed in Subsection~\ref{subsec:methodology}, we also increase the parameterization to enhance learning capacity. In the case of \ac{larpca}, this is achieved through the following design choices:
\begin{enumerate}[label={\em J\arabic*},resume=innovations2] 
    \item \label{itm:sepStep} \textbf{Separate step–sizes for \(\myMat{L}\) and \(\myMat{R}\)} 
          allowing different learning dynamics for the two factors.  
   \item \label{size_matrix}\textbf{Element-wise step-sizes} - To provide additional degrees of freedom for coping with errors induced by approximations, we use matrix-valued step-sizes, such that $\Theta_k = [\bm{\eta}_L^{(k)}, \bm{\eta}_R^{(k)}]$, with $\bm{\eta}_L^{(k)} \in \mathbb{R}^{n_1 \times r}$ and $\bm{\eta}_R^{(k)} \in \mathbb{R}^{n_2 \times r}$. 
\end{enumerate}
Combining \ref{L_approx}-\ref{size_matrix}, the resulting \ac{larpca} is summarized as Algorithm~\ref{alg:LRPCA_iterative}. The formulation returns the low rank component, from which the sparse component is recovered by subtraction.  


\begin{algorithm}
\caption{\ac{larpca}}
\label{alg:LRPCA_iterative}
\SetKwInOut{Initialization}{Init}
    \Initialization{Step-sizes $\{\bm{\eta}_L^{(k)}, \bm{\eta}_R^{(k)}\}$;  
     thresholds $\{\zeta^{(k)}\}$; \newline
    Approximation indices $\mySet{K}_R^{\rm approx}, \mySet{K}_L^{\rm approx}$;   rank $r$ }
    \SetKwInOut{Input}{Input} 
    \Input{Data matrix $\RPCAMat$}  
Set  $\RPCASparse^{(0)} \leftarrow \mySet{T}_{\zeta^{(0)}}(\RPCAMat)$\;
   Compute $[\myMat{U},\myMat{\Sigma},\myMat{V}] \leftarrow \mathrm{SVD}_{r}(\RPCAMat-\RPCASparse^{(0)})$\; 
   Set  $\myMat{L}^{(0)} \leftarrow 
   {{\myMat{U}}}\myMat{\Sigma}^{1/2},\quad
    \myMat{R}^{(0)}\leftarrow \myMat{V}\myMat{\Sigma}^{1/2}$\; 
\For{$k = 0, 1,2 \ldots K-1$}{
    Set    $\RPCASparse^{(k+1)}\leftarrow
        \mySet{T}_{\zeta^{(k+1)}}\!\bigl(\RPCAMat-\myMat{L}^{(k)}(\myMat{R}^{(k)})^T\bigr)$\;
    \If{$k \notin \mySet{K}_L^{\rm approx}$}
    {
         Set $\myMat{G} \!\leftarrow\! \nabla_{\myMat{L}}  \mathcal{L}_{\rm o}(\myMat{L}^{(k)},\myMat{R}^{(k)},\RPCASparse^{(k\!+\!1)}; \RPCAMat)$  \eqref{eq:grad_L}\;
         $\myMat{L}^{(k+1)} \leftarrow \myMat{L}^{(k)}\! - \!{\bm{\eta}}_L^{(k+1)} \odot \myMat{G} ((\myMat{R}^{(k)})^T\myMat{R}^{(k)})^{-1}$\;
    }
    \Else
    {
        Set $\myMat{L}^{(k+1)} \leftarrow \myMat{L}^{(k)}$\;
    }
    \If{$k \notin \mySet{K}_R^{\rm approx}$}
    {
       \hspace{-0.28cm} Set  $\myMat{G}\! \leftarrow\! \nabla_{\myMat{R}}  \mathcal{L}_{\rm o}(\myMat{L}^{(k\!+\!1)},\myMat{R}^{(k)},\RPCASparse^{(k\!+\!1)}; \RPCAMat)$  \hspace{-0.1cm}  via \hspace{-0.05cm}\eqref{eq:grad_R}\;
         $\myMat{R}^{(k+1)}\! \leftarrow \!\myMat{R}^{(k)} \!-\! {\bm{\eta}}_R^{(k\!+\!1)} \odot \myMat{G} ((\myMat{L}^{(k\!+\!1)})^T\myMat{L}^{(k\!+\!1)})^{-1}$\;
    }
    \Else
    {
        Set $\myMat{R}^{(k+1)} \leftarrow \myMat{R}^{(k)}$\;
    }
} 
\KwRet{$\RPCARank^{(K)}=\myMat{L}^{(K)}(\myMat{R}^{(K)})^{T};$}
\end{algorithm} 

                    
                     
                
\ac{larpca} realizes an iterative optimizer with few iterations and approximated low-complexity operations. In particular, the numbers of computed gradients  \eqref{eq:grad_R}, \eqref{eq:grad_L}, and the normalized factors ($(\myMat{R}_k^\top \myMat{R}_k)^{-1}$ and $(\myMat{L}_k^\top \myMat{L}_k)^{-1}$),   decreased by the following formula:
\begin{equation}
    \mathcal{C}_{\rm LARPCA}^{\rm iter}  \!=\! n^{2}r \!+ \!n^{2}\! + \!\frac{2K-|\mySet{K}_R^{\rm approx}|\!-\!|\mySet{K}_L^{\rm approx}|}{K}\mathcal{P}_{\mathrm{low}}.
    \label{eqn:CompRed_rpca}
\end{equation}
When no approximation is applied, i.e., $|\mySet{K}_R^{\rm approx}|=|\mySet{K}_L^{\rm approx}|=0$, then the per-iteration complexity in \eqref{eqn:CompRed_rpca} coincides with \eqref{eqn:rpca_Complexity}.   
The reduction in \eqref{eqn:CompRed_rpca} is shown in Subsection~\ref{ssec:PCA_Exp}  significantly enhance latency, and the increased abstractness  facilitates coping with the mismatches induced by such deliberate approximations.

\subsubsection{Training}
We consider both a supervised and an unsupervised learning setting. In the former, the  data set $\mySet{D}$ used for training consists of realizations of the observation matrix $\RPCAMat$ and the desired low-rank component $\RPCARank^\star$, i.e., 
$\mathcal{D}_{\mathrm{train}}=\{(\RPCAMat^{(i)},\RPCARank^{(i)})\}_{i=1}^{|\mathcal{D}|}$.
Accordingly, training is carried out based on the supervised empirical risk~\eqref{eqn:losssup}. 

For unsupervised data, i.e., when there is no ground-truth low-rank component and one only has access to  $\mathcal{D}_{\mathrm{train}}=\{\RPCAMat^{(i)}\}_{i=1}^{|\mathcal{D}|}$, we formulate a training loss that balances relative reconstruction error and sparsity on $\hat{\RPCASparse}$. Specifically, since the objective in \eqref{eqn:RPCA_nonconvex} does not include an implicit sparse term, we employ the unsupervised loss of \eqref{eqn:lossUnsup} with the following surrogate objective
\begin{equation}
\hat{\mathcal{L}}_{\rm o}(\hat{\RPCASparse}, \hat{\RPCARank}; \RPCAMat) = \frac{\|\RPCAMat -\hat{\RPCARank}\|_F}{\|\RPCAMat\|_F} + \lambda_s \frac{\|\hat{\RPCASparse}\|_1}{n_1\cdot n_2},
\label{eqn:UsupRPCA}
\end{equation}
where $\lambda_s$ is a hyperparameter. 
\vspace{-0.1cm}
\subsection{Experimental Study}\label{ssec:PCA_Exp}
\vspace{-0.1cm}

In this section, we evaluate \ac{larpca} in an experimental study, aiming to demonstrate our ability to achieve comparable results to non-approximated solvers, while reducing computational cost and runtime\footnote{The source code and trained models for both case studies are available at  \url{https://github.com/dviravra/Unfolded_Approximated_Optimization}}. We thus use the following benchmarks
\begin{itemize}  
    \item {\em LRPCA}:  Unfolded \ac{rpca} with learned scalar step-sizes~\cite{cai2021learned}.
    \item {\em Fixed}: Iterative optimization via \eqref{eq:outlier_update}-\eqref{eq:update_R}  with fixed hyperparameters optimized manually to achieve convergence. 
\end{itemize}

To assess accuracy and efficiency, we consider:  
\emph{(i)} Synthetic data simulated from the model in \eqref{eqn:RPCAModel}, where the matrices $\RPCARank^\star$ and $\RPCASparse^\star$ are generated from a Gaussian distribution complying with the low-rank and sparse priors. This synthetic setting lets us isolate the reduction in flops in a controlled environment that fully complies with the modeling assumptions of \ac{rpca}; 
and  
\emph{(ii)} Video decomposition, using the VIRAT video dataset~\cite{oh2011large}, where the low-rank component represents the static aspects of the video, while the sparse component models moving objects. 

\subsubsection{Synthetic Data} 
We first generate an \ac{rpca} setting in which the data matrix sizes are $n_1=n_2=1000$, the low-rank component has $r=5$, and the sparse component has $\alpha=0.1$ non-zero entries. An overall of $|\mySet{D}|=17500$ realizations were used for training the learned optimizers, and for manually tuning the hyperparameters of the {\em fixed} solver to achieve an error level of $10^{-7}$ in reconstructing the low-rank component $\RPCARank^\star$. 
The resulting convergence profile, averaged over $25$ test realizations, is illustrated in Fig.~\ref{fig:PCA_conv}. There, we observe that both learned variants achieve the target error in far fewer iterations compared to the fixed-parameter solver, whose error hardly decreases within the first 25 iterations. It is also observed that the increased abstractness induced by \ref{itm:sepStep}-\ref{size_matrix} allows \ac{larpca} to outperform LRPCA, despite its deliberate occasional approximations. 

\begin{figure}
  \centering
  \includegraphics[width=0.45\textwidth]{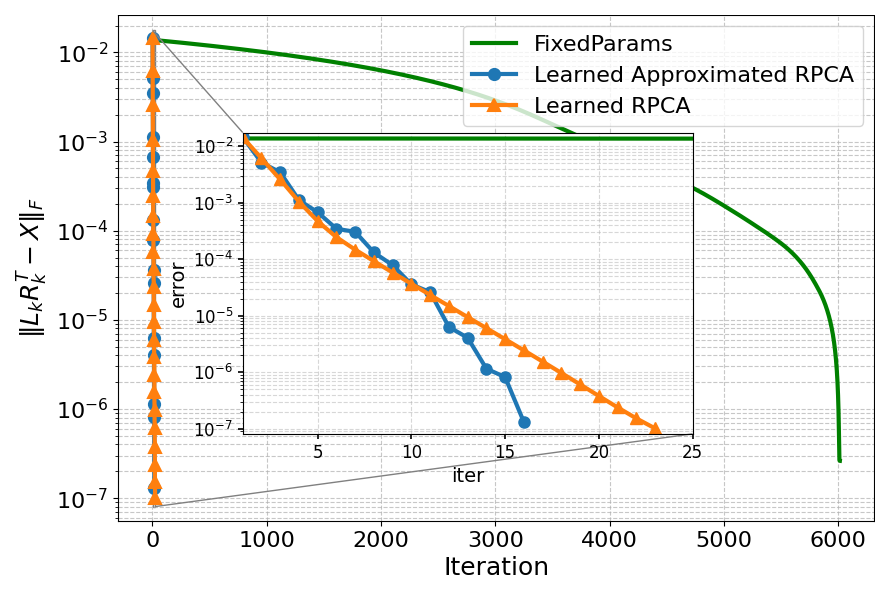}
  \caption{Low-rank recovery error vs. iterations, synthetic data}
  \label{fig:PCA_conv}
\end{figure}


\begin{table}
\centering
 \fontsize{7.5pt}{10pt}\selectfont
\begin{tabular}{|c|c|c|c|c|c|c|}
\hline
 \cellcolor{lightgray}
{\textbf{Method}} &  \cellcolor{lightgray} \bm{$K$} & \cellcolor{lightgray} \textbf{gradient} &
\cellcolor{lightgray} \textbf{flops} \\
\hline
\textbf{\ac{larpca}}& 16 & 16 & $1.768\times10^8$  \\
\hline
\textbf{LRPCA}& 24 & 48 & $3.864\times10^8$  \\
\hline
\textbf{Fixed hyperparameters}& 6000 & 12000 & $9.661\times10^{10}$  \\
\hline
\end{tabular}
\vspace{0.1cm}
\caption{Comparison of optimization methods in terms of computational complexity and performance {with error target $10^{-7}$}.}
\label{tab:PCA_complexity}
\vspace{-0.6cm}
\end{table}

To evaluate the reduction in computations, we run each of the optimizers to  reach a target error of $10^{-7}$. As summarized in Table~\ref{tab:PCA_complexity}, \ac{larpca}  reached this error level in only $16$ iterations, while performing merely $16$ gradient updates in total. The non-approximated unfolded  LRPCA required $24$ iterations (with $48$ gradient updates). This corresponds to skipping 32 out of 48 gradient steps (around $67\%$ reduction). Consequently,  \ac{larpca}  only needed $1.76\times10^{8}$ flops to reach the solution, which is less than half of the $3.86\times10^{8}$ flops required by LRPCA.  To reach the same error of $\approx10^{-7}$, using standard (non-learned) optimization with fixed hyperparameters requires $6000$ iterations and $12000$ gradient calculations, resulting in around $1\times10^{11}$ flops,  which is over three orders of magnitude compared to \ac{larpca}.


\begin{figure}
  \centering
\includegraphics[width=0.9\columnwidth]{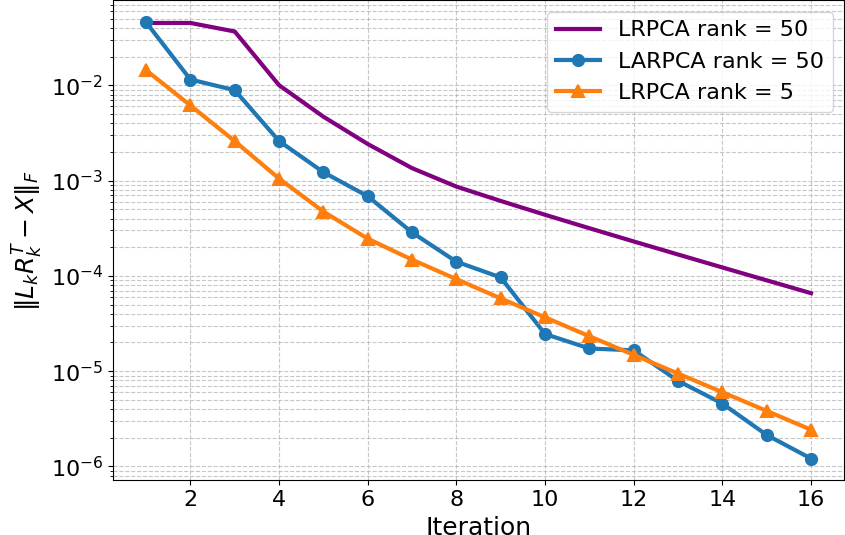}
  \caption{Convergence for higher-rank ($r=50$) vs lower-rank ($r=5$).}
  \label{fig:PCA_high_rank}
  \vspace{-0.2cm}
\end{figure}

As a further evaluation of our proposed method, we examined its robustness to higher-rank scenarios. Fig.~\ref{fig:PCA_high_rank} illustrates a comparison between the convergence profiles of LRPCA and LARPCA for data matrices of increased rank ($r=50$), 
showing that  LRPCA  struggles when the rank of the underlying low-rank component increases. Specifically, when increasing the rank from $r=5$ to $r=50$,  LRPCA  plateaus at a significantly higher error level.  LARPCA, despite  incorporating  approximations (at iterations 6-9), achieves a much lower error, rapidly converging towards the desired solution. This demonstrates that the enhanced flexibility provided by learnable approximations and individual step-size matrices not only reduces computational complexity but also enables RPCA  with challenging, higher-rank data.

\begin{figure*}
  \centering
  \includegraphics[width=0.83\textwidth]{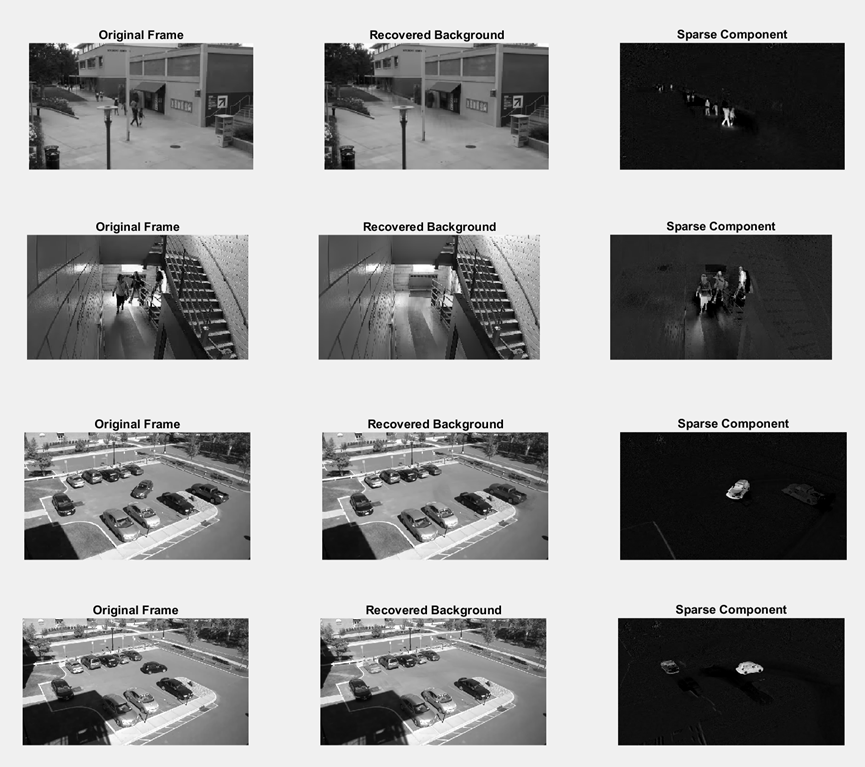}
  \caption{Representative outputs produced by \ac{larpca} for video decomposition.}
  \label{fig:real_data_comparison}
\end{figure*}

We further evaluate the effect of varying the approximation percentage on the reconstruction error. As shown in Table~\ref{table:error_vs_complexity}, multiple simulations were conducted at different approximation levels. The results indicate that even substantial reductions in computational complexity through approximation have only a minor impact on the relative reconstruction error. 

\begin{table}
\centering
\fontsize{7pt}{10pt}\selectfont
\begin{tabular}{|c|c||c|c|}
\hline
\cellcolor{lightgray}\textbf{\% Approximated} & \cellcolor{lightgray}\textbf{Error} & \cellcolor{lightgray}\textbf{\% Approximated} & \cellcolor{lightgray}\textbf{Error} \\
\hline
6.25\%  & $9.07\times10^{-8}$ & 31.25\% & $8.97\times10^{-8}$ \\
\hline
12.50\% & $7.75\times10^{-8}$ & 40\% & $1.31\times10^{-7}$ \\
\hline
25.00\% & $7.75\times10^{-8}$ & 50.00\% & $1.55\times10^{-7}$ \\
\hline
\end{tabular}
\vspace{0.2cm}
\caption{Reconstruction error vs.\ percentage of approximated complexity.}
\label{table:error_vs_complexity}
\vspace{-0.4cm}
\end{table}

\subsubsection{Video Decomposition}
To further validate our learned approximated optimization approach, we evaluated its effectiveness using real-world video data from the VIRAT dataset. This dataset is widely adopted for video surveillance analysis, featuring complex scenes with sparse moving objects over largely static backgrounds, which can be viewed as the low-rank and sparse structure of the RPCA model. For each experiment, we used $16$ videos for training, and $2$ for test. 
 %
%
%
Following the pre-processing approach of \cite{cai2021learned}, we initially convert color frames to grayscale. Subsequently, we uniformly subsampled the frame rate by a factor of two to further reduce data dimensionality and computational complexity, while preserving temporal coherence.
The experimental study comprised sequences from the VIRAT dataset resized to three distinct resolutions, i.e., $640\times360$, $480\times270$, and $320\times180$ pixels.  Each video was vectorized column-wise to form a data matrix suitable for \ac{rpca} decomposition.

\begin{table}
\centering
\fontsize{8pt}{10pt}\selectfont
\renewcommand{\arraystretch}{1.2}
\begin{tabular}{|c|c|c|c|c|}
\hline
\rowcolor{lightgray}
\textbf{Method} & \textbf{Resolution} & \textbf{Frames} & \textbf{Error} & \textbf{Runtime [sec]} \\
\hline
\multirow{4}{*}{\textbf{\ac{larpca}}}
& \multirow{2}{*}{$320\times180$} & 584 & $9.21\times10^{-4}$ & 2.54 \\
\cline{3-5}
& & 750 & $4.16\times10^{-3}$ & 2.67 \\
\cline{2-5}
& $480\times270$ & 2332 & $9.66\times10^{-4}$ & 132.65 \\
\cline{2-5}
& $640\times360$ & 3109 & $1.13\times10^{-3}$ & 1378.21 \\
\hline
\multirow{4}{*}{\textbf{LRPCA}}
& \multirow{2}{*}{$320\times180$} & 584 & $8.72\times10^{-4}$ & 3.91 \\
\cline{3-5}
& & 750 & $4.15\times10^{-3}$ & 5.85 \\
\cline{2-5}
& $480\times270$ & 2332 & $9.66\times10^{-4}$ & 316.11 \\
\cline{2-5}
& $640\times360$ & 3109 & $1.10\times10^{-3}$ & 2193.46 \\
\hline
\end{tabular}
\vspace{0.2cm}
\caption{Comparison of Learned \ac{rpca}  for video decomposition.}
\label{tab:real_data_results}
\vspace{-0.6cm}
\end{table}

The VIRAT dataset is unlabeled for RPCA decomposition, therefore, we employed \eqref{eqn:UsupRPCA} for training the unfolded optimizers.
Due to the computationally intensive nature of this data, particularly at high resolutions, the learned optimizers  operate with $K=5$ iterations, while assuming $r=2$. The reported error is evaluated as the objective  
    $\text{Error}(\hat{\mathbf{V}}, \mathbf{X}) 
    = \frac{\|\mathbf{X} - \hat{\mathbf{V}}\|_F^2}{n_1 \cdot n_2 \cdot \|\mathbf{X}\|_F}$.
All runtime values are reported as wall-clock times on the same platform, a Lenovo IdeaPad Slim~5~14IRL8 with an Intel 13th-Gen Core i5-13420H (8C/12T, 2.10\,GHz base), 16\,GB RAM.
Here, we only evaluate the learned optimizers, as optimization with fixed hyperparameters was unstable and lengthy. 
%
 LARPCA achieves approximately $40\%-58\%$ runtime reduction while preserving the same final accuracy.
Table~\ref{tab:real_data_results} reports quantitative results for representative sequences, clearly indicating  similar reconstruction accuracy (measured via relative squared-error per pixel) compared to  LRPCA, while significantly reducing computational time. We also include four representative results in Fig.~\ref{fig:real_data_comparison}, which demonstrates the effectiveness of \ac{larpca} in foreground-background separation, despite employing approximated computational strategies. 

 \vspace{-0.2cm}
\section{Conclusions}
\label{sec:conclusion} 
We introduced a novel framework for learned approximated optimization, which extends deep unfolding by integrating approximated computations into iterative solvers. 
By deliberately replacing computationally intensive operations with lightweight surrogates and learning to compensate for their induced mismatches, we enable rapid and efficient optimization. We instantiated the proposed methodology in two representative case studies of hybrid beamforming and \ac{rpca}, demonstrating the effectiveness of jointly optimizing iteration depth and per-iteration operations. 

\appendix
\vspace{-0.2cm}
\section{Appendix}
\numberwithin{lemma}{subsection}  
\numberwithin{remark}{subsection} 
\numberwithin{equation}{subsection}	
\subsection{Proof of Proposition~\ref{prop:approx_gd_partial}}
\label{app:proof1}
\vspace{-0.1cm}

 For brevity, let $f(\Label)\triangleq \mathcal{L}_{\rm o}(\Label;\Input)$, $\myVec{g}^{(k)}\triangleq \nabla f(\Label^{(k)})$, and $\myMat{D}_k\triangleq \mathrm{diag}(\bm{\eta}_k)$.
Also define
\[
\myVec{e}^{(k)} \triangleq
\begin{cases}
\tilde{\myVec{g}}^{(k)}(\Label^{(k)};\Input)-\nabla f(\Label^{(k)}), & k\in\mathcal{K}^{\rm approx},\\
\myVec{0}, & k\notin\mathcal{K}^{\rm approx}.
\end{cases}
\]
Then
$\Label^{(k+1)} = \Label^{(k)} - \myMat{D}_k\big(\myVec{g}^{(k)}+\myVec{e}^{(k)}\big)$ and 
$\|\myMat{D}_k \myVec{e}^{(k)}\|_2 \leq \delta_k$.

By the $L$-smoothness of the objective, we have that
\begin{align}
&f(\Label^{(k+1)})
\leq
f(\Label^{(k)})
-
\left\langle \myVec{g}^{(k)}, \myMat{D}_k \myVec{g}^{(k)}\right\rangle
-
\left\langle \myVec{g}^{(k)}, \myMat{D}_k \myVec{e}^{(k)}\right\rangle \notag \\
&\qquad+\frac{L}{2}\|\myMat{D}_k\myVec{g}^{(k)}+\myMat{D}_k\myVec{e}^{(k)}\|_2^2  \notag \\
&= f(\Label^{(k)})
-
\left\langle \myVec{g}^{(k)}, \myMat{D}_k \myVec{g}^{(k)}\right\rangle
-
\left\langle \myVec{g}^{(k)}, \myMat{D}_k \myVec{e}^{(k)}\right\rangle+\frac{L}{2}\|\myMat{D}_k\myVec{g}^{(k)}\|_2^2 \notag \\
&\qquad+
L\left\langle \myMat{D}_k\myVec{g}^{(k)}, \myMat{D}_k\myVec{e}^{(k)}\right\rangle
+\frac{L}{2}\|\myMat{D}_k\myVec{e}^{(k)}\|_2^2 .
\label{eq:proof_start}
\end{align}
 
Next, since all entries of $\bm{\eta}_k$ lie in $[\underline{\eta}_k,\bar{\eta}_k]$, then
\begin{align}
\|\myMat{D}_k\myVec{g}^{(k)}\|_2^2
&=
(\myVec{g}^{(k)})^\top \myMat{D}_k^2 \myVec{g}^{(k)} \leq
\bar{\eta}_k
(\myVec{g}^{(k)})^\top \myMat{D}_k \myVec{g}^{(k)}\notag \\
&
=
\bar{\eta}_k
\left\langle \myVec{g}^{(k)}, \myMat{D}_k \myVec{g}^{(k)}\right\rangle .
\label{eq:D2_bound}
\end{align}
Substituting \eqref{eq:D2_bound} into \eqref{eq:proof_start} gives
\begin{align}
f(\Label^{(k+1)})
&\leq
f(\Label^{(k)})
-
\left(1-\frac{L}{2}\bar{\eta}_k\right)
\left\langle \myVec{g}^{(k)}, \myMat{D}_k \myVec{g}^{(k)}\right\rangle \notag \\
&-
\left\langle \myVec{g}^{(k)}, \myMat{D}_k \myVec{e}^{(k)}\right\rangle
\nonumber\\
&
+
L\left\langle \myMat{D}_k\myVec{g}^{(k)}, \myMat{D}_k\myVec{e}^{(k)}\right\rangle
+\frac{L}{2}\|\myMat{D}_k\myVec{e}^{(k)}\|_2^2 .
\label{eq:after_descent_term}
\end{align}
Applying Cauchy--Schwarz and using $\|\myMat{D}_k \myVec{e}^{(k)}\|_2 \leq \delta_k$ together with
$\|\myMat{D}_k\myVec{g}^{(k)}\|_2 \leq \bar{\eta}_k \|\myVec{g}^{(k)}\|_2$, \eqref{eq:after_descent_term} yields
\begin{align}
f(\Label^{(k+1)})
&\leq
f(\Label^{(k)})
-
\left(1-\frac{L}{2}\bar{\eta}_k\right)
\left\langle \myVec{g}^{(k)}, \myMat{D}_k \myVec{g}^{(k)}\right\rangle \notag \\
&+
(1+L\bar{\eta}_k)\|\myVec{g}^{(k)}\|_2 \delta_k
+\frac{L}{2}\delta_k^2 .
\label{eq:before_young}
\end{align}
Since $\myMat{D}_k \succeq \underline{\eta}_k \myMat{I}$, then $\left\langle \myVec{g}^{(k)},\myMat{D}_k\myVec{g}^{(k)}\right\rangle
\geq
\underline{\eta}_k \|\myVec{g}^{(k)}\|_2^2$ 
and thus
\begin{equation*}
f(\Label^{(k+1)})
\leq
f(\Label^{(k)})
-c_k \|\myVec{g}^{(k)}\|_2^2
+(1+L\bar{\eta}_k)\|\myVec{g}^{(k)}\|_2\,\delta_k
+\frac{L}{2}\delta_k^2 , 
\end{equation*}
where $c_k \triangleq \underline{\eta}_k(1-\frac{L}{2}\bar{\eta}_k)$.
By Young's inequality, we have 
$(1+L\bar{\eta}_k)\|\myVec{g}^{(k)}\|_2\,\delta_k
\leq
\frac{c_k}{2}\|\myVec{g}^{(k)}\|_2^2
+\frac{(1+L\bar{\eta}_k)^2}{2c_k}\delta_k^2 $,
so that
\begin{equation*}
f(\Label^{(k+1)})
\leq
f(\Label^{(k)})
-\frac{c_k}{2}\|\myVec{g}^{(k)}\|_2^2
+
\left(\frac{L}{2}+\frac{(1+L\bar{\eta}_k)^2}{2c_k}\right)\delta_k^2 .
\end{equation*}
The strong convexity ($\|\nabla f(\Label)\|_2^2 \geq 2\mu \big(f(\Label)-f(\Label^\star)\big)$) implies
\begin{equation*}
f(\Label^{(k+1)})-f(\Label^\star)
\leq
(1-\mu c_k)\big(f(\Label^{(k)})-f(\Label^\star)\big)
+
C_k \delta_k^2 ,
, 
\end{equation*}
where $C_k \triangleq \frac{L}{2}+\frac{(1+L\bar{\eta}_k)^2}{2c_k}$.
Since $c_k \geq c$ and $C_k \leq C$, 
\begin{equation}
f(\Label^{(k+1)})-f(\Label^\star)
\leq
(1-\mu c)\big(f(\Label^{(k)})-f(\Label^\star)\big)
+
C\delta_k^2 .
\label{eqn:Recur2}
\end{equation}
Unrolling recursion \eqref{eqn:Recur2} over $k=0,\ldots,K-1$ yields \eqref{eqn:bound_partial}.

\bibliographystyle{IEEEtran} 
\bibliography{IEEEabrv,ref}

\end{document}